\documentclass[11pt]{article}
\usepackage{graphicx}
\usepackage{mathtools,amssymb,amsthm,fullpage}
\usepackage{xcolor}
\usepackage{url}
\usepackage{tikz}
\usepackage{subcaption}
\usepackage[noend]{algpseudocode}
\usepackage{algorithm}
\definecolor{mygreen}{RGB}{34,164,31}
\usepackage{hyperref}
\hypersetup{
    colorlinks=true,
    linkcolor=blue,
    anchorcolor=red,
    citecolor=mygreen,
    filecolor=magenta,
    urlcolor=cyan}
\usepackage{cleveref}

\newcommand{\eps}{\varepsilon}
\newcommand{\davg}{\overline{d}}
\newcommand{\Dyes}{\mathcal{D}^+}
\newcommand{\Dno}{\mathcal{D}^-}
\DeclareMathOperator{\SD}{SD}
\newcommand{\calA}{\mathcal{A}}
\newcommand{\calD}{\mathcal{D}}
\newcommand{\cC}{\mathcal{C}}
\newcommand{\cF}{\mathcal{F}}
\newcommand{\cI}{\mathcal{I}}
\newcommand{\cM}{\mathcal{M}}
\newcommand{\cS}{\mathcal{S}}
\newcommand{\N}{\mathbb{N}}
\newcommand{\E}{\mathop{\mathbb{E}}}
\newcommand{\ceil}[1]{\left\lceil #1 \right\rceil}
\newcommand{\floor}[1]{\left\lfloor #1 \right\rfloor}
\newcommand{\localec}{\textsf{LocalEC}}
\newcommand{\slack}{\mathrm{sl}}
\newcommand{\vol}{\mathrm{vol}}
\newcommand{\dout}{d^{\mathrm{out}}}
\newcommand{\din}{d^{\mathrm{in}}}
\newcommand{\Vlow}{V_{\mathrm{low}}}
\newcommand{\Flow}{\cF_{\mathrm{low}}}
\newcommand{\Fbig}{\cF_{\mathrm{big}}}
\newcommand{\Fgood}{\cF_{\mathrm{good}}}
\newcommand{\rhomax}{\rho_{\max}}
\newcommand{\Gdir}{\overrightarrow{G}}

\newtheorem{theorem}{Theorem}[section]
\newtheorem{claim}[theorem]{Claim}
\newtheorem{lemma}[theorem]{Lemma}
\newtheorem{fact}[theorem]{Fact}
\theoremstyle{definition}
\newtheorem{definition}[theorem]{Definition}
\crefname{claim}{Claim}{Claims}
\crefname{fact}{Fact}{Facts}

\title{Testing Connectedness and Edge Connectedness\footnote{The conference version of this paper will appear at SOSA 2027.}}
\author{
    Diptaksho Palit\\
    Boston University
    \and
    Ramesh Krishnan Pallavoor\\
    Google
    \and 
    Sofya Raskhodnikova\\
    Boston University
    \and
    Nithin Varma\\
    University of Cologne
}

\begin{document}

\maketitle

\begin{abstract}
    We revisit the classical problem of testing connectedness of graphs represented in the adjacency lists model and obtain optimal bounds. Specifically, we show that $\eps$-testing connectedness of graphs with maximum degree at most $d$ requires $\Omega(\frac{1}{\eps} \log \frac{1}{\eps d})$ queries (when $\eps \le \frac{1}{4d^5}$).
    Our result implies that the connectedness tester of Levi, Pallavoor, Raskhodnikova, and Varma ({\em ACM Trans.\ Comput.\ Theory, 2022}) is optimal.
    We also study the problem of testing $k$-edge-connectedness of both directed and undirected graphs of average degree $\overline{d}$ in the adjacency lists model and design a tester with query complexity $O(\frac{k}{\eps} \log \frac{1}{\eps \overline{d}})$, improving upon the bounds by Forster, Nanongkai, Saranurak, Yang, and Yingchareonthawornchai (SODA `20). Furthermore, we complement our upper bound by showing that every one-sided error $k$-edge-connectedness tester must make $\Omega(k/\eps)$ queries.
\end{abstract}

\section{Introduction}\label{sec:intro}

Edge connectedness is a fundamental graph property and is a measure of the fault tolerance of a network that the graph represents. For integer $k \geq 1$, an undirected (directed, respectively) graph $G$ is \emph{$k$-edge-connected}, or simply \emph{$k$-connected}, if at least $k$ edges must be removed from $G$ in order to make the graph disconnected (not strongly connected, respectively). 

The problem of efficiently determining if a graph is $k$-connected has been studied under various computational models and is one of the earliest problems considered in graph property testing~\cite{GR02}.
The goal of an $\eps$-tester for $k$-edge-connectedness is to determine, with probability at least $2/3$, whether a given graph $G = (V,E)$ is $k$-connected, or whether an $\eps \in (0,1)$ fraction of the representation of $G$ needs to be modified in order to obtain a $k$-connected graph.
We consider graphs represented as adjacency lists and algorithms that access the graph via \emph{neighbor} and \emph{degree} queries.
Formally, a neighbor query specifies a vertex $u \in V$ and an integer $i$, upon which an oracle returns the $i$-th (out-)neighbor of $u$, if it exists.
Similarly, a degree query specifies a vertex $u$, upon which an oracle returns the (out-)degree of $u$.

\paragraph{The model.}
We work in the adjacency lists model for general graphs (also called the unbounded-degree incidence-lists model): a graph with $m$ edges is $\eps$-far from a property if at least $\eps m$ edge modifications are needed to obtain a graph with the property, and, as in prior work on testing connectedness in this model~\cite{LeviPRV22,FNSYY20}, the algorithm knows the number of vertices $n$ and the average degree $\davg$ (equivalently, $m$).
This model differs from the bounded-degree model of Goldreich and Ron~\cite{GR02}, in which the algorithm knows an upper bound $d$ on all degrees, distance is measured relative to $nd$ rather than $m$, and the modified graph must also respect the degree bound.
In the general model, degrees can be highly non-uniform, and a tester has to adapt the amount of work it invests to the degrees of the vertices it samples.
For directed graphs, we assume, as in~\cite{OR11,FNSYY20}, that both outgoing and incoming adjacency lists (and degrees) can be queried; if only outgoing lists can be queried, testing strong connectedness requires $\Omega(\sqrt{n})$ queries~\cite{BenderR02}.

\paragraph{State of the art.}
For general $k$, the best known testers are due to Forster, Nanongkai, Saranurak, Yang, and Yingchareonthawornchai~\cite{FNSYY20}, for both directed and undirected graphs.
In the general model, their testers make $\tilde{O}(k/\eps^2)$ queries, and $\tilde{O}(\min\{k/\eps^2, 1/\eps^3\})$ queries for simple directed graphs, even when $\davg$ is unknown~\cite[Theorem~1.3]{FNSYY20}.
When $\davg$ is known, as in our setting, their bounds are $\tilde{O}(k^2/(\eps^2\davg))$ and $\tilde{O}(\min\{k^2/(\eps^2 \davg),\, k/(\eps^3\davg)\})$, respectively~\cite[Theorem~6.1]{FNSYY20}; these are never worse, since one may then assume $\davg\ge k$.
In the bounded-degree model, their testers make $\tilde{O}(k/\eps)$ queries~\cite[Theorem~1.4]{FNSYY20}.
These bounds were an exponential improvement, in terms of the dependence on $k$, over prior results~\cite{YoshidaI10,OR11}.
For the important special case of undirected connectedness ($k = 1$), the best known tester has query complexity $O\left( \min \left\{ \frac{1}{(\eps \overline{d})^2}, \frac{1}{\eps} \log \frac{1}{\eps \overline{d}} \right\} \right)$~\cite{LeviPRV22}.

\subsection{Our results}

Despite improvements on the upper bound side, little is known about the tightness of the obtained bounds for connectedness testing.
Our first result is a matching lower bound for testing connectedness.
We show that every tester for connectedness of graphs with maximum degree $d$ and average degree $\davg = \frac{d}{2}+1$ requires
$\Omega\!\left(\frac1\eps\log\frac1{\eps\davg}\right)$
queries when $\eps \le \frac{1}{4d^5}$ (\Cref{thm:conn-test-lb}).
In this range of parameters, $\frac{1}{\eps}\log\frac{1}{\eps\davg} = O\big(\frac{1}{(\eps\davg)^2}\big)$, so our bound establishes the optimality of the tester of Levi, Pallavoor, Raskhodnikova, and Varma~\cite{LeviPRV22}.

Our second contribution concerns testing $k$-edge-connectedness.
Building on a local subroutine of Forster, Nanongkai, Saranurak, Yang, and Yingchareonthawornchai~\cite{FNSYY20}, we simplify their algorithm and obtain a tester with query complexity
$O\!\left(\frac{k}{\eps}\log\frac1{\eps\davg}\right)$
for both directed and undirected simple graphs (\Cref{thm:dir-k-connectedness,thm:k-connectedness}).
In the model we study, where $\davg$ is known, the best previous bound for simple graphs is $\tilde{O}\big(\min\{\frac{k^2}{\eps^2\davg}, \frac{k}{\eps^3\davg}\}\big)$~\cite[Theorem~6.1]{FNSYY20}, so our bound is smaller by a factor of $\tilde{\Omega}\big(\min\{\frac{k}{\eps\davg}, \frac{1}{\eps^2\davg}\}\big)$ whenever $\eps\davg<1$.
Moreover, the case $\eps\davg\ge 4$, which our bounds do not explicitly address, is easy: for simple graphs, it suffices to check the degrees of $O\big(\frac{k}{\eps\davg}\big)$ uniformly random vertices (see \Cref{app:large-eps}).

Finally, we complement this upper bound by proving that every one-sided error tester for $k$-edge-connectedness requires $\Omega(k/\eps)$ queries, for graphs of average degree $\davg \ge k+2$ and $\eps \le \frac{1}{32(\davg+1)^2}$ (\Cref{thm:one-sided-lb}).

\subsection{Prior work}\label{sec:prior-work}

Goldreich and Ron~\cite{GR02} gave the first testers for edge and vertex connectedness in the bounded-degree model, initiating a long line of work on sublinear-time algorithms for connectedness-related properties.
For general $k$, the current state of the art is due to Forster, Nanongkai, Saranurak, Yang, and Yingchareonthawornchai~\cite{FNSYY20}, who designed testers for both directed and undirected edge and vertex connectedness with $\tilde{O}(k/\eps)$ queries in the bounded-degree model, and with the bounds stated above in the general model.
Their work substantially improved the dependence on the connectedness parameter over previous algorithms~\cite{GR02,YoshidaI10,YoshidaI12,OR11}.

For the special case of connectedness ($k=1$), Parnas and Ron~\cite{PR02} gave the first sublinear tester in the adjacency list model.
This line of work was subsequently refined by Berman, Raskhodnikova, and Yaroslavtsev~\cite{BermanRY14} and culminated in the
$O\left( \min \left\{ \frac{1}{(\eps \overline{d})^2}, \frac{1}{\eps} \log \frac{1}{\eps \overline{d}} \right\} \right)$-query tester of Levi, Pallavoor, Raskhodnikova, and Varma~\cite{LeviPRV22}.

\subsection{Technical overview}

Our lower bound for testing connectedness is based on a graph primitive, denoted $G_S$, consisting of a cycle of \emph{supernodes}, where each supernode is a $d$-clique with a pendant edge attached to all but two of its vertices.
The key property of $G_S$ is that, to move from a supernode to the next one, a tester has to find the single edge leaving the supernode among $\Theta(d^2)$ adjacency list entries; hence exploring $s$ supernodes costs $\Omega(sd^2)$ queries in expectation, which is a constant fraction of the number of edges in the explored part of the graph.
Our hard distributions consist of a connected instance formed by a single copy of $G_S$ (for the yes case), and disconnected instances formed by many smaller copies of $G_S$ (for the no case), where the size of the small cycles is chosen uniformly at random from a logarithmic range.
The graph primitive prevents the tester from discovering a cycle with significantly fewer queries than its size, while the random choice of cycle size reflects the logarithmic factor from the work investment argument of~\cite{BermanRY14,LeviPRV22}, yielding our $\Omega\left(\frac{1}{\eps}\log\frac{1}{\eps d}\right)$ lower bound.

Our tester for $k$-connectedness builds on the framework of~\cite{FNSYY20}.
Specifically, we use their \localec\ subroutine and observe that, for the purpose of property testing, its volume and slack parameters can be merged into a single parameter.
Combined with the work investment strategy and a separate treatment of vertices of degree less than $k$, this simplification leads to both a cleaner analysis and an improved query complexity.

Finally, our one-sided error lower bound for testing $k$-connectedness is inspired by the connectedness lower bound.
Here, we consider a single cycle size and modify the supernode construction so that every component is itself $k$-connected.
A one-sided error tester can reject only if the answers to its queries are inconsistent with every $k$-connected graph.
We show that, unless the tester makes $\Omega(k/\eps)$ queries, the edges it has not seen can be rewired so as to make the graph $k$-connected without changing any of the answers, leading to the desired lower bound.

\section{Preliminaries}\label{sec:preliminaries}

Throughout the paper, all graphs are simple, and $\log$ denotes the logarithm to base $2$.
For an integer $t \ge 1$, we write $[t] = \{1,\ldots,t\}$.
We assume that the number of vertices $n$ and the average degree $\davg$ of the input graph are known to the algorithm; for an undirected graph $G=(V,E)$, we have $\davg = 2|E|/n$, and for a directed graph, $\davg = |E|/n$ is the average out-degree.
Graphs are represented in the adjacency list model, where the algorithm may make \emph{degree} and \emph{neighbor} queries.
For directed graphs, we assume access to both incoming and outgoing degree and neighbor queries.

The distance of a graph from $k$-edge-connectedness is measured by the number of edge modifications (insertions and deletions) required to obtain a $k$-edge-connected graph.
Specifically, an undirected graph $G=(V,E)$ is \emph{$\eps$-far} from $k$-edge-connectedness if at least $\eps |E|$ edge modifications are required to make $G$ $k$-edge-connected.
The definition for directed graphs is analogous.

For a directed graph $G=(V,E)$ and $S \subseteq V$, let $E(S,V\setminus S)$ denote the set of edges from $S$ to $V\setminus S$, and let $\dout(S) = |E(S,V\setminus S)|$ and $\din(S)=|E(V\setminus S,S)|$.
For a vertex $v$, we write $\dout(v)$ and $\din(v)$ for its out- and in-degree, and we let $\vol(S)=\sum_{u\in S}\dout(u)$ denote the \emph{out-volume} of $S$.
A directed graph is $k$-edge-connected if and only if $\dout(S)\ge k$ for every set $S$ with $\emptyset\neq S\subsetneq V$; an undirected graph is $k$-edge-connected if and only if at least $k$ edges cross every such cut $(S,V\setminus S)$.

\begin{definition}[Tester for $k$-connectedness]
    A randomized algorithm $\mathcal{A}$ is an \emph{$\eps$-tester} for $k$-edge-connectedness if, given query access to a graph $G$, it satisfies the following:
    \begin{itemize}
        \item if $G$ is $k$-edge-connected, then $\Pr[\mathcal{A}\text{ accepts }G]\ge \frac23$;
        \item if $G$ is $\eps$-far from $k$-edge-connectedness, then $\Pr[\mathcal{A}\text{ rejects }G]\ge \frac23$.
    \end{itemize}
    A tester has \emph{one-sided error} if it always accepts $k$-edge-connected graphs.
\end{definition}

Our upper bounds are stated for the interesting range $\eps\davg\le \frac12$, in which $\log\frac{1}{\eps\davg}\geq 1$.\footnote{For larger $\eps$, all our upper bounds hold with $\log\frac{1}{\eps\davg}$ replaced by $1$.}

\section{Lower bound for testing connectedness}\label{sec:conn-lb}

In this section, we prove a lower bound for testing connectedness of graphs that are represented as adjacency lists.
Our lower bound implies that the connectedness tester of Levi et al.~\cite{LeviPRV22} has optimal query complexity when $\eps \le \frac{1}{4d^5}$.
A tester is allowed to make both neighbor queries and degree queries.

\begin{theorem}\label{thm:conn-test-lb}
    For $d \ge 2$ and $\eps \in \left(0, \frac{1}{4d^3}\right]$, every algorithm that $\eps$-tests connectedness on graphs of average degree $d$ in the adjacency list model with degree and neighbor queries requires $\Omega\left(\frac{1}{\eps}\log\frac{1}{\eps \davg}\right)$ queries for sufficiently large inputs\footnote{Our constructions work for $n \ge \frac{1}{\eps^4}$; we have not optimized this.}.
\end{theorem}

\begin{proof}
    We use Yao's minimax principle~\cite{Yao77} as stated by~\cite{RS06}.
    
    \begin{theorem}[Yao's minimax principle, the version from~\cite{RS06}]\label{thm:yao}
        To prove a lower bound $q$ on the worst-case query complexity of a randomized algorithm, it is enough to give
        two distributions on inputs:
        \begin{itemize}
            \item $\Dyes$ on positive instances, and
            \item $\Dno$ on negative instances
        \end{itemize}
        such that it is hard for any $q$-query deterministic algorithm $\calA$ to
        distinguish $\Dyes$ from $\Dno$.
        Let $a_1\dots a_q(G)$ be the answers to the queries on input $G$. Define $\Dyes$-view to be the
        distribution on the subgraphs of $G$ induced by $a_1\dots a_q(G)$ when $G$ is selected according to
        $\Dyes$. Similarly, define $\Dno$-view. By ``hard to distinguish $\Dyes$ from $\Dno$'', we mean
        $\SD\left(\text{$\Dyes$-view},\text{$\Dno$-view}\right)<\frac 1 3$, where $\SD$ denotes the statistical distance between two distributions.
    \end{theorem}
    
    Specifically, we construct distributions $\Dyes$ and $\Dno$ over $n$-vertex graphs that are connected and $\eps$-far from connected, respectively, such that every deterministic algorithm $\calA$ that makes at most $q = c\cdot \frac1\eps \log\frac1{\eps d^3}$ queries, for a sufficiently small constant $c > 0$, fails to distinguish the two distributions with high probability.

    \paragraph{Graph primitives.}
    We use the following graph primitives for our constructions.

    \begin{definition}[Graph $S_d$]\label{def:spider}
        The {\em spider graph $S_d$,} parameterized by degree $d\in\N$, where $d\geq 2$, has $2d-2$ vertices, partitioned into sets $\{ v_1, \dots, v_d \}$ and $\{ u_1, \dots, u_{d-2} \}$.
        The first $d$ vertices, called {\em internal}, form a $d$-clique, and there is an {\em exit edge} between $u_i$ and $v_i$ for all $i \in [d-2]$.
        Vertices $u_1,\ldots,u_{d-2}$ are called {\em leaves}; vertices $v_{d-1}$ and $v_{d}$ are called the \emph{source} and \emph{terminal}, respectively.
    \end{definition}

    \begin{definition}[Graph $G_S(n', d)$]\label{def:gs}
        Let $d,n'\in\N$ be such that $n'\ge 4d-4$ is a multiple of $2d-2$.
        The {\em graph $G_S(n', d)$} is a cycle consisting of $n'/(2d-2)$ {\em supernodes}, where each supernode is a copy of the spider graph $S_d$, with the {\em source} of one supernode adjacent to the {\em terminal} of the next (see \Cref{fig:gs}).
        The edge between the terminal of one supernode and the source of the next one is also called an {\em exit edge}, and the source and the terminal of a supernode are called its {\em ports}.
    \end{definition}

    The graph $G_S(n',d)$ has $n'$ vertices, out of which $\frac{n' d}{2d-2}$ internal vertices have degree $d$ each and $\frac{n' (d-2)}{2d-2}$ leaves have degree $1$ each.
    So the maximum degree is $d$, the average degree is $\davg = \frac{d}{2} + 1$, and the number of edges is $\frac{n'(d+2)}{4}$.
    We will prove a query lower bound of $\Omega\left(\frac1\eps\log\frac{1}{\eps\davg}\right)$ for average degree $\davg = \frac{d}{2} + 1$. 
    This suffices because substituting $d' = 2(d-1)$ does not change the asymptotic bound.

    \begin{figure}[h]
        \centering
        \tikzset{
            intv/.style={circle, draw, fill=white, inner sep=1pt},
            portv/.style={circle, draw=green!40!black, fill=green!50, inner sep=1pt},
            leafv/.style={circle, draw, fill=black!15, inner sep=1pt},
            cliquee/.style={draw=black!40},
            leafexit/.style={draw=red!85!black, line width=0.9pt},
            portexit/.style={draw=blue!80!black, line width=1.2pt}
        }
        \begin{subfigure}{0.4\textwidth}
            \centering
            \begin{tikzpicture}[scale=2]
                \def\R{0.6}
                \def\RR{1.1}
                \foreach \i in {1,...,4} \node[intv] (v\i) at ({60*(\i-1)}:\R) {$v_{\i}$};
                \node[portv] (v5) at (240:\R) {$v_5$};
                \node[portv] (v6) at (300:\R) {$v_6$};
                \foreach \i in {1,...,6} {
                    \foreach \j in {1,...,6} {
                        \ifnum\i<\j
                            \draw[cliquee] (v\i) -- (v\j);
                        \fi
                    }
                }
                \foreach \i in {1,...,4} {
                    \node[leafv] (u\i) at ({60*(\i-1)}:\RR) {$u_{\i}$};
                    \draw[leafexit] (v\i) -- (u\i);
                }
                \draw[portexit, dashed] (v5) -- (240:1.05);
                \draw[portexit, dashed] (v6) -- (300:1.05);
                \node[below] at (240:1.05) {\strut source};
                \node[below] at (300:1.05) {\strut terminal};
            \end{tikzpicture}
            \caption{The spider graph $S_6$.}
            \label{fig:spider}
        \end{subfigure}
        \hfill
        \begin{subfigure}{0.56\textwidth}
            \centering
            \begin{tikzpicture}[scale=1.6]
                \def\n{8}
                \def\R{0.5}
                \def\RR{0.7}
                \def\BIG{2.0}
                \foreach \i in {1,...,\n} {
                    \pgfmathsetmacro{\ang}{360/\n*(\i-1)}
                    \begin{scope}[shift={({\ang}:\BIG)}, rotate=\ang+90]
                        \foreach \k in {1,...,4} \node[intv, inner sep=1.4pt] (v\i_\k) at ({-60*(\k-1)}:\R) {};
                        \node[portv, inner sep=1.4pt] (v\i_5) at (-240:\R) {};
                        \node[portv, inner sep=1.4pt] (v\i_6) at (-300:\R) {};
                        \foreach \k in {1,...,6} {
                            \foreach \m in {1,...,6} {
                                \ifnum\k<\m
                                    \draw[cliquee] (v\i_\k) -- (v\i_\m);
                                \fi
                            }
                        }
                        \foreach \k in {1,...,4} {
                            \node[leafv, inner sep=1.4pt] (u\i_\k) at ({-60*(\k-1)}:\RR) {};
                            \draw[leafexit] (v\i_\k) -- (u\i_\k);
                        }
                    \end{scope}
                }
                \foreach \i in {1,...,\n} {
                    \pgfmathtruncatemacro{\next}{mod(\i,\n)+1}
                    \draw[portexit] (v\i_6) -- (v\next_5);
                }
            \end{tikzpicture}
            \caption{The graph $G_S(80, 6)$, with $8$ supernodes.}
            \label{fig:gs}
        \end{subfigure}
        \caption{The graph primitives used in the proof of \Cref{thm:conn-test-lb}.
        The ports of each supernode (its source and terminal) are green, the other internal vertices are white, and the leaves are light gray.
        Exit edges between an internal vertex and a leaf are red, exit edges between the ports of consecutive supernodes are blue, and the edges of the $d$-cliques are gray.}
        \label{fig:primitives}
    \end{figure}
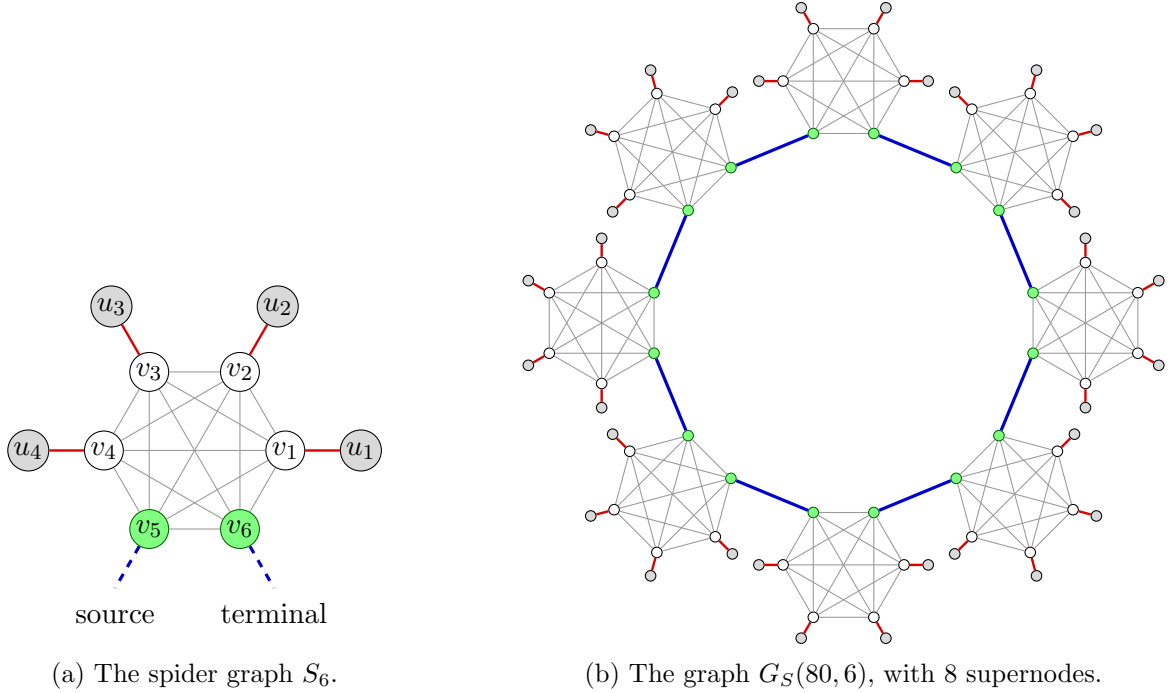

    \paragraph{Hard distributions.}
    Let $J = \{ j \in \N : d^2 \le 2^j \le \frac{1}{\eps d}\}$ and $L = |J|$.
    For $j\in J$, let $n_j = (2d-2)\floor{\frac{2^j}{2d-2}}$ and $S_j = \frac{n_j}{2d-2}$, so that $G_S(n_j,d)$ has $S_j$ supernodes.
    Let $t = \ceil{\frac{\eps n d}{2}}$.
    Distribution $\Dno$ first chooses $j^\star$ uniformly at random from $J$, and returns a uniformly random isomorphism of a graph consisting of
    \begin{itemize}
        \item $t$ connected components isomorphic to $G_S(n_{j^\star},d)$, called \emph{small cycles}, and
        \item one connected component isomorphic to $G_S\left(n - t\cdot n_{j^\star},d\right)$, called the \emph{large cycle}.
    \end{itemize}
    Distribution $\Dyes$ outputs a uniformly random isomorphism of $G_S(n,d)$.
    In both distributions, the order of every adjacency list is an independent, uniformly random permutation.

    Every graph in $\Dyes$ is connected, and every graph in $\Dno$ is $\eps$-far from connected, since it has $t+1$ connected components, so at least $t \ge \frac{\eps n d}{2} \ge \eps\cdot\frac{n(d+2)}{4}$ edges must be added to connect it.
    We will show that $\calA$ cannot distinguish these two distributions.
    
    Let $N = \frac{n}{2d-2}$ denote the number of supernodes.
    Since $\eps \le \frac{1}{4d^3}$, we have $\log\frac{1}{\eps d^3} \geq 2$, and hence $L \ge \frac12\log\frac{1}{\eps d^3}$.
    

    \paragraph{Views of the algorithm.}
    We make the following simplifying assumption on the view of the algorithm $\calA$: whenever a vertex of a supernode appears in the view for the first time, 
    the labels of all $d$ {\em internal} vertices of the supernode are {\em revealed} to $\calA$, and if this vertex is a leaf, then its unique neighbor is revealed as well.
    The positions of the vertices in the adjacency lists are not revealed.
    A vertex is {\em discovered} if it has appeared in the view, and a supernode is {\em discovered} if at least one of its vertices has appeared in the view.
    A query is \emph{fresh} if it is made to a vertex that is not discovered.
    We consider two discovered supernodes adjacent if an exit edge between them has been revealed, and define the \emph{connected components of the view} accordingly.

    We say that $\calA$ \emph{resolves} an internal vertex $w$ when it queries the adjacency list of $w$ at the position of the exit edge of $w$.
    If $w$ is a port, then resolving $w$ reveals a vertex of the neighboring supernode; we call this an \emph{extension}.
    Since all internal vertices of a discovered supernode are known, the only way for $\calA$ to reach a new supernode is a fresh query or an extension.

    \paragraph{The idealized process and the coupling.}
    We compare both distributions to the following process $\cI$, which answers the queries of $\calA$ as if the input consisted of infinitely many supernodes such that no two explored parts are ever connected.
    Whenever a supernode has to be revealed, because of a fresh query or an extension, $\cI$ creates a new supernode: its internal vertices receive distinct labels chosen uniformly at random among the labels that are not discovered. 
    Note that the run of $\cI$ does not depend on $\Dyes$, $\Dno$, or $j^\star$.

    For $\calD\in\{\Dyes,\Dno\}$, let $P_{\calD}$ denote the process that answers the queries of $\calA$ on $G\sim\calD$, where $G$ is sampled lazily.
    Since $G$ is a uniformly random isomorphism of a fixed graph, in a fresh query of $P_{\calD}$ the queried vertex is uniformly distributed among the undiscovered vertices of $G$.
    Such a vertex is either a leaf of a discovered supernode (we call this a \emph{collision}), or a uniformly random vertex of an undiscovered supernode $X$.
    In the latter case, and also when an extension of $P_{\calD}$ reaches an undiscovered supernode, the new supernode is revealed exactly as in $\cI$. 
    Therefore, we can couple $P_{\calD}$ with $\cI$ as follows.
    We run $\cI$ and, in parallel, sample $G$ lazily: we sample $j^\star$ if $\calD = \Dno$.
    Whenever $\cI$ creates a new connected component of the view by a fresh query, first decide whether this query results in a collision in $P_{\calD}$.
    If not, choose the undiscovered supernode $X$ of $G$ to which the new supernode of $\cI$ is mapped; and every extension of $\cI$ is mapped to the corresponding neighboring supernode in $G$.
    Under this coupling, $P_{\calD}$ and $\cI$ give identical answers until the first query at which a collision occurs or an extension of $P_{\calD}$ reaches a supernode that is already discovered.
    Note that $j^\star$ and the supernodes of $G$ at which the connected components of the view start are chosen using randomness that is independent of the run of $\cI$.

    \paragraph{Bad events.}
    We now define the bad events.
    The first and second are events of the coupled process for $\calD\in\{\Dyes,\Dno\}$, and together they describe the first query at which $P_{\calD}$ and $\cI$ give different answers; the third is an event of the run of $\cI$ only.
    \begin{description}
        \item[$\mathsf{Merge}$ ($\cM_{\calD}$):]
        Before $\mathsf{SmallCycle}$ occurs, a fresh query of $P_{\calD}$ results in a collision, or an extension of $P_{\calD}$ reaches a discovered supernode in a different connected component of the view (that is, two components of the view merge).
        
        \item[$\mathsf{SmallCycle}$ ($\cS_{\calD}$):]
        Before $\mathsf{Merge}$ occurs, an extension of $P_{\calD}$ reaches a discovered supernode in the same connected component of the view (that is, the algorithm discovers a cycle of supernodes).

        \item[$\mathsf{CheapExplore}$ ($\cC$):]
        During the run of $\cI$, the total number of extensions is at least $\frac{q}{c_1 d^2}$, for a sufficiently small constant $c_1 > 0$.
        Since a connected component of the view of $\cI$ with $s$ supernodes is discovered using $s-1$ extensions, this means that the algorithm discovers connected components with $s_1,\ldots,s_K$ supernodes while making at most $q \le c_1 d^2\sum_{i=1}^K (s_i - 1)$ queries.
    \end{description}

    \paragraph{Indistinguishability assuming no bad events.}
    By the coupling, the views of $\calA$ under $P_{\calD}$ and under $\cI$ are identical unless $\cM_{\calD}\cup\cS_{\calD}$ occurs.
    Recall that if two random variables can be coupled so that they are equal except on an event of probability $p$, then their statistical distance is at most $p$.
    Hence $\SD(\calD\text{-view}, \cI\text{-view}) \le \Pr[\cM_{\calD}\cup\cS_{\calD}]$ for both $\calD\in\{\Dyes,\Dno\}$, where $\cI$-view denotes the view of $\calA$ under $\cI$.
    Note that $\cS_{\Dyes}$ never occurs, since closing the only cycle of supernodes of a graph from $\Dyes$ requires $N > q$ extensions.
    For the same reason, $\cS_{\Dno}$ can occur only on a small cycle, as the large cycle has more than $q$ supernodes.

    \paragraph{Probability of bad events.}
    Now it remains to bound the probability of the bad events.
    
    \begin{lemma}\label{lem:merge}
        For both $\calD\in\{\Dyes,\Dno\}$, we have $\Pr[\cM_{\calD}] \le \frac{1}{21}$.
    \end{lemma}
    
    \begin{proof}
        At any point during the execution, at most $q$ supernodes and at most $(d+1) q \le 2dq$ vertices are discovered, since each query reveals at most one new supernode and at most one leaf.
 
        \emph{Collisions.}
        In a fresh query, the queried vertex is uniform among at least $n-2dq$ undiscovered vertices, of which at most $(d-2)q$ are leaves of discovered supernodes.
        By a union bound over at most $q$ fresh queries, a collision occurs with probability at most $\frac{dq^2}{n-2dq}$.
 
        \emph{Merges through extensions.}
        Consider a fresh query that reveals a new supernode $X$, which is uniform among at least $N-q$ undiscovered supernodes.
        Each extension moves one step along a cycle, and $\calA$ makes at most $q$ extensions, so the component of $X$ can later run into an earlier component only if $X$ is within distance $q$ of a supernode discovered before $X$.
        At most $q$ supernodes are discovered before $X$, and at most $q(2q+1)\le 3q^2$ supernodes are within distance $q$ of them, so this happens with probability at most $\frac{3q^2}{N-q}$.
        Charging each merge to the fresh query that started the later of the two components, a union bound over at most $q$ fresh queries shows that two components merge with probability at most $\frac{3q^3}{N-q}\le\frac{6dq^3}{n-2dq}$, where we used $N-q\ge\frac{n-2dq}{2d}$.
 
        Summing these for large enough $n$, we get $\Pr[\cM_{\calD}] \le \frac{7dq^3}{n-2dq} \le \frac{50}{49}\cdot\frac{7dq^3}{n} \le \frac{50}{49}\cdot\frac{7}{200} < \frac{1}{21}$.
    \end{proof}
    
    \begin{lemma}\label{lem:cheap}
        $\Pr[\cC] \le \frac{1}{21}$.
    \end{lemma}
    
    \begin{proof}
        We first compute the expected number of extensions.
        Call an internal vertex of a discovered supernode $X$ a \emph{candidate} if $\calA$ does not yet know whether it is a port.
        When $X$ is discovered, it has at least $d-1$ candidates, at most two of which are ports.
        The exit edge of an internal vertex $w$ is uniform among its $d$ neighbors, so if $f_w$ neighbors of $w$ have been queried, the next query resolves $w$ with probability $\frac{1}{d-f_w}$.
        The ports are also uniform among the $d$ internal vertices of $X$, so if $r_X$ candidates have been resolved, the next resolved candidate is a port with probability at most $\frac{2}{d - 1 -r_X}$.
        An extension is a query that resolves a candidate that is a port. We split the extensions according to the state at the time of the query.
        \begin{itemize}
            \item \emph{Cheap extensions}, with $f_w<\frac d2$ and $r_X<\frac{d-1}{2}$.
            Each query is a cheap extension with probability at most $\frac2d\cdot\frac{8}{d}$, so the expected number of cheap extensions is at most $\frac{16q}{d^2}$.
            \item \emph{Extensions at heavily queried vertices}, with $f_w\ge\frac d2$ and $r_X<\frac{d-1}{2}$.
            At most $\frac{2q}{d}$ vertices receive more than $\frac d2$ queries, and each of them is resolved at most once, as a port with probability at most $\frac8d$.
            So the expected number of these extensions is at most $\frac{16q}{d^2}$.
            \item \emph{Extensions in heavily explored supernodes}, with $r_X\ge\frac{d-1}{2}\ge\frac d4$.
            Split the total number of resolutions to vertices which have at most $\frac{d}{2}$ neighbors queried, and those without.
            The expected number of resolutions with few queries is at most $\frac{2}{d} \cdot q$, and the resolutions with a lot of queries can be bounded above by $\frac{2q}{d}$.
            Thus the expected total number of resolutions is at most $\frac2d\cdot q+\frac{2q}{d}=\frac{4q}{d}$, so by the law of total expectation, the expected number of supernodes with at least $\frac d4$ resolved candidates is at most $\frac{16q}{d^2}$.
            Each has at most two extensions, so the expected number of these extensions is at most $\frac{32q}{d^2}$.
        \end{itemize}
        The expected number of extensions is thus at most $\frac{64q}{d^2}$, and by Markov's inequality, $\Pr[\cC]\le \frac1{21}$.
    \end{proof}

    \begin{lemma}\label{lem:smallcycle}
        $\Pr\big[\cS_{\Dno} \cap \overline{\cC}\big] \le \frac{1}{21}$, and thus
        $\Pr\big[\cS_{\Dno}\big] \le \frac{2}{21}$.
    \end{lemma}
    
    \begin{proof}
        We first bound the probability that a component in the view starts (upon discovering a fresh supernode $X$) on a small cycle.
        Fix $j\in J$, and condition on $j^\star = j$, on the run of $\cI$, and on the choices made for the components created before component $i$.
        The supernode $X$ is uniformly distributed among at least $N-q$ undiscovered supernodes, of which at most $t\cdot S_j$ lie on small cycles.
        The probability that $X_i$ lies on a small cycle is at most
        \[
            \frac{t \cdot S_j}{N-q} \;\le\; \frac{50}{49}\cdot\frac{\eps n d}{2}\cdot\frac{2^j}{2d-2}\cdot\frac{1}{N-q} \;=\; \frac{50}{49}\cdot\eps d\, 2^{j-1}\cdot\frac{N}{N-q} \;\le\; \left(\frac{50}{49}\right)^2 \eps d\, 2^{j-1}.
        \]
        Importantly, this bound holds for every fixed run of $\cI$, because the coupling chooses $j^\star$ and $X_i$ independent to the run of $\cI$.
        In particular, we never need to analyze how the choices of $\calA$ depend on $j^\star$: the view under $\cI$ does not depend on $j^\star$ at all.
        Next, we condition on $\overline{\cC}$.
        This is an event of the run of $\cI$, so it does not change the distribution of $j^\star$ or of the $X_i$.

        For each component $i$ in the view, let $e_i$ denote the number of extensions in that component.
        Taking a union bound over all components and averaging over the choice of $j^\star$, which is uniform over the $L$ elements of $J$, we get
        \begin{align*}
            \Pr\big[\cS_{\Dno}\cap\overline{\cC}\big] \;\le\; \Pr\big[\cS_{\Dno}\mid\overline{\cC}\big]
            &\le \sum_{j\in J}\frac{1}{L}\sum_i \left(\frac{50}{49}\right)^2 \eps d\, 2^{j-1}\cdot \Pr\big[e_i\ge S_j \bigm| \overline{\cC}\big]\\*
            &= \frac{1}{L}\left(\frac{50}{49}\right)^2 \eps d\cdot \E\Big[\sum_i \sum_{j\in J:\, S_j\le e_i} 2^{j-1} \Bigm| \overline{\cC}\Big].
        \end{align*}
        For each component $i$, the inner sum is over cycle sizes that component $i$ is long enough to traverse.
        By definition, $\frac{2^j}{4d-4} \le S_j$, so
        \[
            \sum_{j\in J:\, S_j\le e_i} 2^{j-1} \;\le\; \sum_{j\in\N:\, 2^j\le (4d-4)e_i} 2^{j-1} \;<\; (4d-4)e_i \;\le\; 4d\, e_i.
        \]
        Conditioned on $\overline{\cC}$, we have $\sum_i e_i < \frac{q}{c_1 d^2}$.
        Therefore,
        \[
            \Pr\big[\cS_{\Dno}\cap\overline{\cC}\big] \;\le\; \frac{1}{L}\left(\frac{50}{49}\right)^2 \eps d\cdot 4d\cdot\frac{q}{c_1d^2} \;=\; \left(\frac{50}{49}\right)^2 \frac{4\eps q}{c_1 L} \;\le\; \left(\frac{50}{49}\right)^2\frac{8c}{c_1} \;<\; \frac{1}{21}. \qedhere
        \]
    \end{proof}

    \paragraph{Completing the proof.}
    By the triangle inequality and \Cref{lem:merge,lem:smallcycle},
    \begin{align*}
        \SD\left(\Dyes\text{-view},\Dno\text{-view}\right)
        &\le \SD\left(\Dyes\text{-view},\cI\text{-view}\right) + \SD\left(\cI\text{-view},\Dno\text{-view}\right)\\
        &\le \Pr[\cM_{\Dyes}] + \Pr[\cM_{\Dno}] + \Pr[\cS_{\Dno}] \;\le\; \frac{1}{21}+\frac{1}{21}+\frac{2}{21} \;=\; \frac{4}{21} \;<\; \frac13 .
    \end{align*}
    By Yao's principle, every $\eps$-tester for connectedness of the graphs in our distributions makes more than $q = c\cdot\frac1\eps\log\frac{1}{\eps d^3}$ queries.
    Because $\eps \le \frac{1}{4d^5}$, we get $\frac{1}{\eps d}\ge 4d^4$, so $\log\frac{1}{\eps d^3} = \log\frac{1}{\eps d} - 2\log d \ge \frac12\log\frac{1}{\eps d}$.
    Moreover, $\davg \ge \frac{d}{2}$ and $\log\frac{1}{\eps d}\ge 1$ give $\log\frac{1}{\eps\davg} \le \log\frac{1}{\eps d}+1 \le 2\log\frac{1}{\eps d}$. 
    Hence $q = \Omega\left(\frac{1}{\eps}\log\frac{1}{\eps\davg}\right)$.
\end{proof}

\section{$k$-Connectedness Tester}\label{sec:tester}
In this section, we present our tester for $k$-connectedness.
We first consider the directed setting, where we instantiate the local subroutine of~\cite{FNSYY20} with parameters tailored to our application, yielding an improved query complexity.
We then show a reduction from testing undirected $k$-connectedness to testing directed $k$-connectedness, allowing us to obtain the desired bounds in the undirected setting.

\subsection{Directed $k$-Connectedness Tester}

In this section, we prove the following upper bound for testing $k$-connectedness of directed graphs.

\begin{theorem}\label{thm:dir-k-connectedness}
    There exists a one-sided error $\eps$-tester for $k$-edge-connectedness of simple directed graphs that makes 
    $O\left(\frac{k}{\eps} \log \frac{1}{\eps \davg} \right)$
    queries (with access to both incoming and outgoing adjacency lists).
    Its running time is 
    $O\left(\frac{k^2}{\eps} \log \frac{1}{\eps \davg} \right)$ if the input graph has $m\ge\frac{520k}{\eps}$ edges, and $O(k^2/\eps^2)$ otherwise.
\end{theorem}

We use the following characterization of graphs that are far from being $k$-connected.

\begin{fact}[{\cite[Corollary~8]{OR11}}]\label{fact:far-characterization}
    A directed graph $G=(V,E)$ is $\eps$-far from being $k$-edge-connected if and only if there exists a family of disjoint non-empty proper subsets $\cF = \{S_1,\ldots,S_t\} \subseteq V$ satisfying $\sum_{i \in [t]} (k - d^{\text{out}}(S_i)) > \eps m$ or $\sum_{i \in [t]} (k - d^{\text{in}}(S_i)) > \eps m$.
\end{fact}

Note that the sets $S\in\cF$ with $\dout(S)\ge k$ can be removed from $\cF$ without decreasing the first sum (and similarly for the second sum).

We use a local algorithm of~\cite{FNSYY20} as a subroutine.
The following lemma is a corollary of \cite{FNSYY20}, but we include a proof in \Cref{app:localec} for completeness.
Specifically, \cite{FNSYY20} required four parameters for their algorithm, but by a careful analysis, we merge two of them.

\begin{lemma}\label{lem:localec}
    There is an algorithm $\localec$ that takes as inputs
    \begin{itemize}
        \item a vertex $v$ of an $m$-edge directed graph $G=(V,E)$, accessed via out-degree and out-neighbor queries,
        \item a cut-size parameter $k \in \mathbb{Z}^+$, and
        \item a \emph{volume-slack ratio} parameter $\rho \in \mathbb{R}$, where $1 \leq \rho \le \frac{m}{130k}$.
    \end{itemize}
    It makes $O(k \rho)$ queries, runs in $O(k^2 \rho)$ time, and outputs either $\bot$ or a set $S$ with $\emptyset\neq S\subsetneq V$ and $\dout(S) < k$.
    Moreover, if there exist $\gamma \in [k]$ and a vertex set $S' \subseteq V$ such that $v \in S'$, $\vol(S') \leq \rho \gamma$, and $\dout(S') \leq k - \gamma$, then $\localec$ outputs a set (rather than $\bot$) with probability at least $3/4$.
\end{lemma}

We also use the following lemma of Berman, Raskhodnikova, and Yaroslavtsev~\cite{BermanRY14}, which underlies their work investment strategy.

\begin{lemma}[{\cite[Lemma~2.3]{BermanRY14}}]\label{lem:wis}
    Let $X$ be a random variable that takes values in $[0,1]$. Suppose $\E[X] \ge \beta>0$, let $t = \ceil{\log \frac{4}{\beta}}$, and let $\delta \in (0,1)$. For all $i \in [t]$, let $p_i = \Pr[X \ge 2^{-i}]$ and $s_i = \ceil{\frac{4\ln (1/\delta)}{2^i\beta}}$. Then $\prod_{i=1}^{t} (1 - p_i)^{s_i} \le \delta$.
\end{lemma}


We now describe our tester. Let $\rhomax = \frac{m}{130k}$, where $m = n\davg$ is known to the algorithm.

\begin{algorithm}[h]
	\caption{Directed $k$-Connectedness Tester}
	\begin{algorithmic}[1]
		\Require $\eps \in (0,1)$, $k\ge 1$, $n$, $\davg$, oracle access to a simple directed graph $G$ via in- and out-degree and neighbor queries
		\If{$\davg < k$} \textbf{reject} \label{step:avg-deg}
		\EndIf
		\If{$\eps m < 520k$} query all out-adjacency lists of $G$; \textbf{reject} if $G$ is not $k$-edge-connected, and \textbf{accept} otherwise \label{step:small}
		\EndIf
		\For{$H \in \{G, G^R\}$, where $G^R$ is $G$ with all edges reversed}\label{step:directions}
		\Loop~$s_0 = \ceil{\frac{4k \ln 20}{\eps \davg}}$ times: Sample $v \in V$ uniformly at random; \textbf{reject} if $d^{\mathrm{out}}_H(v) < k$\label{step:low-sampling}
		\EndLoop
		\For {$i = 1$ to $\ceil{\log\left(\frac{8}{\eps \davg}\right)} $}\label{step:for-loop-dir-k-conn}
		\Loop~$s_i = \ceil{\frac{8 \ln 20}{2^i \eps \davg}}$ times:\label{step:repetition-dir-k-conn}
		\State Sample $v \in V$ uniformly at random and query $d^{\mathrm{out}}_H(v)$; \textbf{reject} if $d^{\mathrm{out}}_H(v) < k$\label{step:vertex-sample-dir-k-conn}
		\State Run $\localec(v,k,\min\{2^i \cdot d^{\mathrm{out}}_H(v), \rhomax \})$ on $H$ and \textbf{reject} if it returns a set\label{step:dir-FNSYY}
		\EndLoop
		\EndFor
		\EndFor
		\State \textbf{Accept}
	\end{algorithmic}
\label{alg:dir-k-connectedness}
\end{algorithm}

\begin{proof}[Proof of \Cref{thm:dir-k-connectedness}]
    \emph{One-sided error.}
    If $G$ is $k$-edge-connected, then every vertex has out-degree and in-degree at least $k$, so $\davg \ge k$, and every set $S$ with $\emptyset \neq S \subsetneq V$ satisfies $\dout(S)\ge k$ and $\din(S)\ge k$ (note that $\din(S)$ is the out-degree of $S$ in $G^R$).
    Hence Step~\ref{step:avg-deg} does not reject, Step~\ref{step:small} accepts if it applies, no sampled vertex has out-degree less than $k$ in $H$, and, by \Cref{lem:localec}, \localec\ never returns a set.
    Thus, Algorithm~\ref{alg:dir-k-connectedness} always accepts graphs that are $k$-connected.

    \emph{Rejection of far graphs.}
    Consider a graph $G=(V,E)$ that is $\eps$-far from $k$-connected.
    If Step~\ref{step:avg-deg} or Step~\ref{step:small} applies, then the algorithm rejects with probability $1$. So assume that $\davg\ge k$ and $\eps m\ge 520k$; in particular, $\rhomax \ge 4$.
    By \Cref{fact:far-characterization}, and since the algorithm handles both $G$ and $G^R$ (where $d^{\mathrm{out}}_{G^R}(S) = \din(S)$), we may assume without loss of generality that there exists a family $\cF$ of pairwise disjoint nonempty proper subsets of $V$ such that $\sum_{S\in\cF} (k-\dout(S)) > \eps m$, and consider the iteration of Step~\ref{step:directions} with $H=G$.
    Fix such a family $\cF$, and remove from it all sets $S$ with $\dout(S) \geq k$.

    \begin{definition}[Slack] \label{def:cut-set-slack}
        For a set $S \in \cF$, the \emph{slack} of $S$ is $\slack(S) := k - \dout(S)$.
    \end{definition}

    For each $S \in \cF,$ the slack satisfies $1 \le \slack(S) \leq k$, and $\sum_{S\in\cF}\slack(S) > \eps m$.
    We partition $\cF$ into three parts. Let $\Vlow = \{v \in V : \dout(v) < k\}$ be the set of \emph{low-degree} vertices, and let
    \begin{itemize}
        \item $\Flow = \{S\in\cF : S\cap \Vlow \neq\emptyset\}$,
        \item $\Fbig = \{S\in\cF\setminus\Flow : \vol(S) > \rhomax\cdot \slack(S)\}$, and
        \item $\Fgood = \cF\setminus(\Flow\cup\Fbig)$.
    \end{itemize}
    Since the sets in $\Fbig$ are disjoint, $\sum_{S\in\Fbig}\slack(S) < \sum_{S\in\Fbig}\frac{\vol(S)}{\rhomax}\le \frac{m}{\rhomax} = 130k \le \frac{\eps m}{4}$.
    Hence $\sum_{S\in\Flow}\slack(S) + \sum_{S\in\Fgood}\slack(S) > \frac{3\eps m}{4}$, and one of the following two cases holds.

    \emph{Case 1: $\sum_{S\in\Flow}\slack(S) \geq \frac{\eps m}{4}$.}
    Since $\slack(S)\le k$, we have $|\Flow|\ge\frac{\eps m}{4k}$, and since the sets in $\Flow$ are disjoint and each of them contains a low-degree vertex, $|\Vlow| \geq \frac{\eps m}{4k}$.
    So each vertex sampled in Step~\ref{step:low-sampling} is a low-degree vertex with probability at least $\frac{\eps m}{4kn} = \frac{\eps\davg}{4k}$, and the algorithm fails to reject in this step with probability at most $\left(1-\frac{\eps\davg}{4k}\right)^{s_0} \le e^{-\ln 20} = \frac{1}{20}$.

    \emph{Case 2: $\sum_{S\in\Fgood}\slack(S) > \frac{\eps m}{2}$.}
    We define the \emph{$k$-quality} of the vertices with respect to $\Fgood$.

    \begin{definition}[$k$-quality]\label{def:k-quality}
        For a vertex $v \in V$, its \emph{$k$-quality}, denoted by $q_k(v)$, is defined as
        \[
            q_k(v) = 
            \begin{cases}
                \frac{\dout(v)}{\vol(S)} \cdot \slack(S) & \text{ if } v \in S \text{ for some } S \in \Fgood,\\
                0 & \text{ otherwise.}
            \end{cases}
        \]
    \end{definition}
    This is well defined, because every vertex of $S \in \Fgood$ has $\dout \ge k$, so $\vol(S) \ge k|S| > 0$.

    \begin{claim}\label{clm:exp-k-quality}
        For every vertex $v\in V$, we have $0\le q_k(v) < 1$. Moreover, $\E_{v \in V} \left[ q_k(v) \right] > \frac{\eps \davg}{2}$.
    \end{claim}
    
    \begin{proof}
        Let $S\in\Fgood$, $s = |S|$, and $\sigma = \slack(S)$, so that $\dout(S) = k-\sigma \le k-1$.
        Every vertex of $S$ has out-degree at least $k$, and at most $s-1$ of its out-edges stay inside $S$.
        So every vertex of $S$ has at least $k-s+1$ out-edges leaving $S$, and $\dout(S)\geq s(k-s+1)$.
        If $s\le k$, then $s(k-s+1) \ge k$, a contradiction.
        Thus $s \ge k+1$.
        Fix $v \in S$.
        Each out-edge of $v$ goes to one of the $s-1$ other vertices of $S$ or leaves $S$, so $\dout(v)\le s-1+\dout(S) = s-1+k-\sigma$. Together with $\vol(S)\ge ks$, we get
        \[
            q_k(v) \;\le\; \frac{(s-1+k-\sigma)\,\sigma}{ks} \;\le\; \frac{(s-1)\,k}{ks} \;=\; \frac{s-1}{s} \;<\; 1,
        \]
        where the second inequality holds because $(s-1+k-\sigma)\sigma$ is increasing in $\sigma$ for $\sigma \le \frac{s-1+k}{2}$, and $\sigma\le k\le\frac{s-1+k}{2}$.
        
        Now, $\sum_{v \in S} q_k(v) = \sum_{v \in S} \frac{\dout(v)}{\vol(S)} \cdot \slack(S) = \slack(S)$ for every $S\in\Fgood$.
        Since the sets are disjoint,
        \[
            \sum_{v \in V} q_k(v) = \sum_{S \in \Fgood} \slack(S) > \frac{\eps m}{2} = \frac{\eps\davg n}{2},
        \]
        and the bound on $\E_{v \in V} \left[ q_k(v) \right]$ follows.
    \end{proof}

    Let $X$ be the random variable denoting the $k$-quality of a vertex sampled in Step~\ref{step:vertex-sample-dir-k-conn}.
    By \Cref{clm:exp-k-quality}, $X$ takes values in $[0,1]$ and $\E[X] > \beta := \frac{\eps\davg}{2}$, and the parameters in Steps~\ref{step:for-loop-dir-k-conn}--\ref{step:repetition-dir-k-conn} are exactly $t = \ceil{\log\frac{4}{\beta}}$ and $s_i = \ceil{\frac{4\ln 20}{2^i\beta}}$.
    Hence, by \Cref{lem:wis} with $\delta = \frac{1}{20}$, with probability at least $\frac{19}{20}$, there is an $i \in [t]$ such that Algorithm~\ref{alg:dir-k-connectedness} samples a vertex $v$ with $q_k(v) \geq 2^{-i}$ in iteration $i$.
    Fix such $i$ and $v$.
    Since $q_k(v) > 0$, there is a set $S \in \Fgood$ with $v \in S$, and $\gamma := \slack(S) \in [k]$ satisfies $\dout(S) = k-\gamma$.
    Moreover, $\dout(v)\ge k$, so the algorithm runs $\localec$ with $\rho = \min\{2^i \dout(v), \rhomax\}$. 
    By definition of $q_k$, we have $\vol(S) \le 2^i \dout(v)\cdot\gamma$ and also $S\notin\Fbig$, so $\vol(S) \le \rhomax \cdot \gamma$, meaning $\vol(S)\le \rho\gamma$.
    By \Cref{lem:localec}, \localec\ returns a set, and the algorithm rejects, with probability at least $3/4$.
    So in Case~2, the algorithm rejects with probability at least $\frac{19}{20}\cdot\frac{3}{4} > \frac{71}{100}$.

    \emph{Query complexity and running time.}
    Step~\ref{step:avg-deg} makes no queries.
    If Step~\ref{step:small} applies, then $\eps m < 520k$, and since $\davg\ge k$, we have $n \le m/k \le m$. So reading all adjacency lists takes $n + m = O(k/\eps)$ queries.
    Moreover, $G$ is $k$-edge-connected if and only if, for a fixed vertex $r$ and every other vertex $u$, there are $k$ edge-disjoint paths from $r$ to $u$ and from $u$ to $r$. This can be checked with at most $k$ augmenting path computations for each of the $2(n-1)$ pairs, in time $O(nk(n+m)) = O(m^2) = O(k^2/\eps^2)$, using $n\le m/k$.
    Otherwise, Step~\ref{step:low-sampling} makes $O(s_0) = O\big(\frac{k}{\eps\davg}\big) = O\big(\frac{1}{\eps}\big)$ queries, since $\davg\ge k$.
    By \Cref{lem:localec}, for a vertex $v$ sampled in iteration $i$, \localec\ makes $O(k \cdot 2^i \cdot d^{\mathrm{out}}_H(v))$ queries and runs in $O(k^2\cdot 2^i\cdot d^{\mathrm{out}}_H(v))$ time.
    A uniformly random vertex has expected out-degree $\davg$, both in $G$ and in $G^R$.
    Since $s_i \le \frac{8\ln 20}{2^i\eps\davg}+1$, $2^t \le \frac{16}{\eps\davg}$, and $\eps\davg\le\frac12$, the expected query complexity of Steps~\ref{step:for-loop-dir-k-conn}--\ref{step:dir-FNSYY} is at most
    \[
        \sum_{i=1}^{t} s_i\cdot O\left(k\cdot 2^i \cdot \davg\right) = \sum_{i=1}^{t} O\left(\frac{k}{\eps} + k\,2^i\davg\right) = O\left(\frac{k t}{\eps} + k\,2^{t}\davg\right) = O\left( \frac{k}{\eps} \log\frac{1}{\eps \davg} \right).
    \]
    By a similar calculation, the expected running time is $O \left(\frac{k^2}{\eps} \log\frac{1}{\eps \davg}\right)$.
    
    The tester claimed in \Cref{thm:dir-k-connectedness} is obtained by running Algorithm~\ref{alg:dir-k-connectedness}, halting and accepting if the number of queries or the running time exceeds 50 times the corresponding bound on its expectation established above. This modification preserves one-sided error.
    By Markov's inequality and a union bound, the algorithm is halted with probability at most $\frac{1}{25}$.
    In both cases above, the unmodified algorithm accepts an $\eps$-far graph with probability at most $\frac{29}{100}$, so the modified algorithm accepts it with probability at most $\frac{1}{25} + \frac{29}{100} < \frac{1}{3}$.
    This completes the proof of \Cref{thm:dir-k-connectedness}. 
\end{proof}

\subsection{Undirected $k$-Connectedness Tester}

We first show how to reduce undirected $k$-connectedness testing to directed $k$-connectedness testing, and then use \Cref{thm:dir-k-connectedness} to derive a tester for undirected graphs.

\begin{definition}[Reduction]
    Given an undirected graph $G=(V,E)$, define its directed version $\Gdir$ by replacing each edge $\{u,v\} \in E$ with two directed edges $(u,v)$ and $(v,u)$.
\end{definition}

Note that, for every $S\subseteq V$, the edges of $\Gdir$ from $S$ to $V\setminus S$ are in one-to-one correspondence with the edges of $G$ between $S$ and $V\setminus S$.

\begin{lemma}[Reduction completeness]\label{lem:reduction-completeness}
    If an undirected graph $G$ is $k$-connected, then its directed version $\Gdir$ is also $k$-connected.
\end{lemma}

\begin{proof}
    We prove the contrapositive.
    If $\Gdir$ is not $k$-connected, then there is a set $S$ with $\emptyset\neq S\subsetneq V$ such that fewer than $k$ edges of $\Gdir$ go from $S$ to $V\setminus S$.
    By the correspondence above, fewer than $k$ edges of $G$ cross the cut $(S, V\setminus S)$, and removing them disconnects $G$. So $G$ is not $k$-connected.
\end{proof}

\begin{lemma}[Reduction soundness]\label{lem:reduction-soundness}
    If an undirected graph $G$ is $\eps$-far from being $k$-connected, then its directed version $\Gdir$ is $\eps/2$-far from being $k$-connected.
\end{lemma}

\begin{proof}
    Let $\overrightarrow{H}$ be a $k$-connected directed graph that is obtained from $\Gdir$ with $M$ edge modifications, and let $A$ be the set of edges of $\overrightarrow{H}$ that are not edges of $\Gdir$, so that $|A|\le M$.
    The graph obtained by adding $A$ to $\Gdir$ contains $\overrightarrow{H}$, so it is also $k$-connected.
    Let $G'$ be the graph obtained from $G$ by inserting the edge $\{u,v\}$ for every pair $\{u,v\}$ such that $(u,v)\in A$ or $(v,u)\in A$; this inserts at most $|A|$ edges, none of which belongs to $G$.
    For every set $S$ with $\emptyset\neq S\subsetneq V$, distinct edges of $\Gdir$ or of $A$ going from $S$ to $V\setminus S$ correspond to distinct edges of $G'$ crossing the cut $(S,V\setminus S)$. Hence at least $k$ edges of $G'$ cross this cut, and $G'$ is $k$-connected.
    Since $G$ is $\eps$-far from $k$-connected, $M\ge |A| \ge \eps|E| = \frac{\eps}{2}\cdot 2|E|$. So $\Gdir$ is $\eps/2$-far from being $k$-connected.
\end{proof}

We are now ready to state our final tester for $k$-connectedness of undirected graphs.

\begin{theorem}\label{thm:k-connectedness}
    There exists a one-sided error $\eps$-tester for $k$-edge-connectedness of simple undirected graphs with average degree $\davg$ that makes $O\left(\frac{k}{\eps} \log \frac{1}{\eps \davg} \right)$ queries.
    Its running time is $O\left(\frac{k^2}{\eps} \log \frac{1}{\eps \davg} \right)$ if the input graph has at least $\frac{520k}{\eps}$ edges, and $O(k^2/\eps^2)$ otherwise.
\end{theorem}

\begin{proof}
    Given an undirected graph $G$, run the tester from \Cref{thm:dir-k-connectedness} on $\Gdir$ with parameter $\eps/2$.
    Each query to $\Gdir$ (to an incoming or outgoing adjacency list or degree) can be answered with one query to $G$, and $\Gdir$ has the same average degree as $G$ and twice as many edges, so the query and time complexities follow from \Cref{thm:dir-k-connectedness}.
    If $G$ is $k$-connected, then $\Gdir$ is $k$-connected by \Cref{lem:reduction-completeness}, so the tester always accepts.
    If $G$ is $\eps$-far from $k$-connected, then, by \Cref{lem:reduction-soundness}, $\Gdir$ is $\eps/2$-far from $k$-connected, so the tester rejects with probability at least $2/3$.
\end{proof}

\section{One-sided error lower bound for $k$-connectedness testing}\label{sec:one-sided-lb}

We prove the following lower bound for one-sided error testers of $k$-edge-connectedness.

\begin{theorem}\label{thm:one-sided-lb}
    Let $k\ge 1$ and $d\ge k+3$ be integers, let $\eps\in\left(0,\frac{1}{32d^2}\right]$, and let $n\ge\frac{k}{\eps d}$ be a multiple of $d$.
    Every one-sided error $\eps$-tester for $k$-edge-connectedness of $n$-vertex undirected graphs with average degree $\davg = d-1$, in the adjacency list model with degree and neighbor queries, makes at least $\frac{k}{32\eps}$ queries on some input.
\end{theorem}

This is consistent with the $\tilde O(\min\{k/\eps, 1/\eps^2\})$ upper bound of~\cite{FNSYY20} for simple graphs in the bounded-degree model, since our assumptions imply $\eps\le\frac{1}{32(k+3)^2}$, and therefore $\frac{k}{\eps} \le \frac{1}{\eps^2}$.

\begin{proof}[Proof of \Cref{thm:one-sided-lb}]
    Let $h = \ceil{k/2}$, and note that $2h\le k+1 < d$.

    \paragraph{Graph primitives.}
    We use the following graph primitives for our construction.
    
    \begin{definition}[$T_{d, k}$]
        The graph $T_{d, k}$ has $d$ vertices labeled $1, \dots, d$. It is obtained from the $d$-clique on these vertices by removing the $h$ edges of the matching $\{1, 2\}, \{3, 4\}, \dots, \{2h-1, 2h\}$.
    \end{definition}
    
    \begin{definition}[$G_T (N', d, k)$]
        For an integer $N'\ge 2$, the graph $G_T (N', d, k)$ is a cycle consisting of $N'$ supernodes, indexed by $\mathbb{Z}_{N'}$, where each supernode is a copy of $T_{d, k}$.
        For every $i\in\mathbb{Z}_{N'}$ and $\ell\in[h]$, vertex $2\ell-1$ of supernode $i$ is joined by an edge to vertex $2\ell$ of supernode $i+1$.
        The edges of the copies of $T_{d,k}$ are called \emph{internal} edges.
    \end{definition}

    In $T_{d,k}$, the vertices $1,\ldots,2h$ have degree $d-2$, and all other vertices have degree $d-1$. Each of the vertices $1,\ldots,2h$ of a supernode is incident to exactly one edge between supernodes, so $G_T(N',d,k)$ is a simple $(d-1)$-regular graph with $N'd$ vertices.

    \begin{figure}[h]
        \centering
        \tikzset{
            intv/.style={circle, draw, fill=white, inner sep=1pt},
            matchv/.style={circle, draw=green!40!black, fill=green!50, inner sep=1pt},
            cliquee/.style={draw=black!40},
            removede/.style={draw=red!85!black, densely dashed, line width=0.9pt},
            unseene/.style={draw=red!85!black, line width=0.9pt},
            betweene/.style={draw=blue!80!black, line width=1.2pt}
        }
        \begin{subfigure}{0.4\textwidth}
            \centering
            \begin{tikzpicture}[scale=2.3]
                \def\R{0.6}
                \node[matchv] (t1) at (30:\R) {$1$};
                \node[matchv] (t2) at (150:\R) {$2$};
                \node[matchv] (t3) at (330:\R) {$3$};
                \node[matchv] (t4) at (210:\R) {$4$};
                \node[intv] (t5) at (90:\R) {$5$};
                \node[intv] (t6) at (270:\R) {$6$};
                \foreach \p/\q in {1/3,1/4,1/5,1/6,2/3,2/4,2/5,2/6,3/5,3/6,4/5,4/6,5/6} \draw[cliquee] (t\p) -- (t\q);
                \draw[removede] (t1) -- (t2);
                \draw[removede] (t3) -- (t4);
                \draw[betweene, dashed] (t1) -- ++(0.45,0);
                \draw[betweene, dashed] (t3) -- ++(0.45,0);
                \draw[betweene, dashed] (t2) -- ++(-0.45,0);
                \draw[betweene, dashed] (t4) -- ++(-0.45,0);
                \node[font=\footnotesize] at (1.15,0) {$i{+}1$};
                \node[font=\footnotesize] at (-1.15,0) {$i{-}1$};
            \end{tikzpicture}
            \caption{The supernode $T_{6,3}$.}
            \label{fig:tdk}
        \end{subfigure}
        \hfill
        \begin{subfigure}{0.56\textwidth}
            \centering
            \begin{tikzpicture}[scale=1.6]
                \def\n{8}
                \def\R{0.42}
                \def\BIG{2.0}
                \foreach \i in {1,...,\n} {
                    \pgfmathsetmacro{\ang}{360/\n*(\i-1)}
                    \begin{scope}[shift={({\ang}:\BIG)}, rotate=\ang+90]
                        \node[matchv, inner sep=1.4pt] (g\i_1) at (30:\R) {};
                        \node[matchv, inner sep=1.4pt] (g\i_2) at (150:\R) {};
                        \node[matchv, inner sep=1.4pt] (g\i_3) at (330:\R) {};
                        \node[matchv, inner sep=1.4pt] (g\i_4) at (210:\R) {};
                        \node[intv, inner sep=1.4pt] (g\i_5) at (90:\R) {};
                        \node[intv, inner sep=1.4pt] (g\i_6) at (270:\R) {};
                        \foreach \p/\q in {1/3,1/4,1/5,1/6,2/3,2/4,2/5,2/6,3/5,3/6,4/5,4/6,5/6} \draw[cliquee] (g\i_\p) -- (g\i_\q);
                    \end{scope}
                }
                \foreach \i in {1,...,\n} {
                    \pgfmathtruncatemacro{\next}{mod(\i,\n)+1}
                    \draw[betweene] (g\i_1) -- (g\next_2);
                    \draw[betweene] (g\i_3) -- (g\next_4);
                }
            \end{tikzpicture}
            \caption{The graph $G_T(8,6,3)$, with $8$ supernodes.}
            \label{fig:gt}
        \end{subfigure}
        \caption{The graph primitives used in the proof of \Cref{thm:one-sided-lb}, for $d=6$ and $k=3$, so that $h=2$.
        The vertices $1,\ldots,2h$ of each supernode (green) are incident to the edges between supernodes (blue); the other vertices are white, the edges of the $d$-cliques are gray, and the removed matching edges are dashed red.
        In (a), the dashed blue stubs indicate the edges to the neighboring supernodes: the odd vertices $1$ and $3$ are adjacent to the next supernode $i+1$, and the even vertices $2$ and $4$ to the previous supernode $i-1$.}
        \label{fig:one-sided}
    \end{figure}
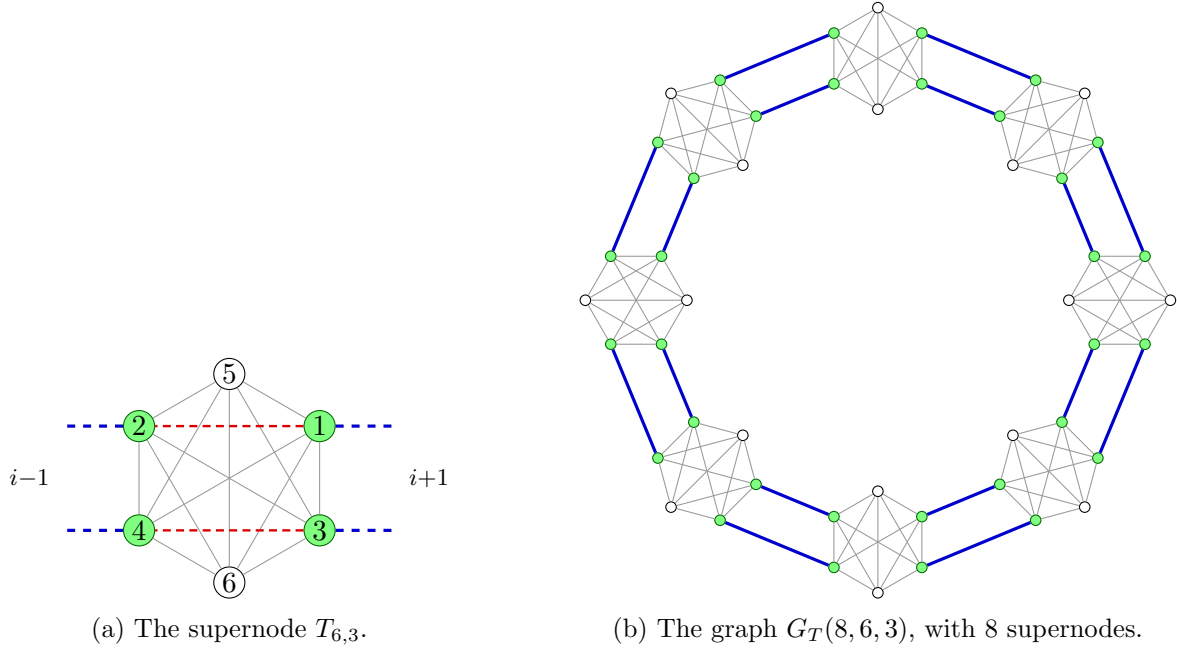

    \begin{claim}\label{clm:g_t-is-k-conn}
        The graph $G_T(N', d, k)$ is $k$-connected.
        Moreover, a graph $G'$ obtained by removing at most one internal edge from each supernode of $G_T(N', d, k)$ is also $k$-connected.
    \end{claim}

    \begin{proof}
        Consider a cut $(A,B)$ of $G'$ that separates two vertices of some supernode $X$, and let $a = |X\cap A|$, where $1\le a\le d-1$.
        Exactly $a(d-a)$ edges of the $d$-clique on $X$ cross the cut. At most $\min\{a, d-a, h\}$ of them belong to the removed matching, and at most one more was removed from $X$ to obtain $G'$.
        If $a\in\{1,d-1\}$, then at least $(d-1)-1-1 = d-3 \ge k$ internal edges of $X$ cross the cut.
        Otherwise, $a(d-a)\ge 2(d-2)$, and at least $2(d-2)-h-1 \ge 2(k+1)-\frac{k+1}{2}-1 \geq k$ internal edges of $X$ cross the cut.
        
        Now consider a cut $(A,B)$ that separates two supernodes, but no two vertices of the same supernode. Then $A$ and $B$ are unions of supernodes. Since the supernodes are arranged in a cycle, there are at least two indices $i\in\mathbb{Z}_{N'}$ such that supernodes $i$ and $i+1$ are on different sides of the cut. For each such $i$, the $h$ edges between supernodes $i$ and $i+1$ cross the cut. So at least $2h \ge k$ edges cross the cut.
        Thus $G'$ and subsequently $G_T(N',d,k)$ is $k$-connected.
    \end{proof}

    \paragraph{Hard instance.}
    Let $N_0 = \floor{\frac{k}{4\eps d^2}}$ and $r = \ceil{\frac{\eps n d}{k}}$.
    Let $G$ be the graph consisting of $r$ connected components isomorphic to $G_T(N_0,d,k)$, called \emph{small components}, and one connected component isomorphic to $G_T(N_1,d,k)$, called the \emph{large component}, where $N_1 = \frac{n}{d} - rN_0$. 
    Thus $G$ is a simple $(d-1)$-regular graph with $n$ vertices and $\frac{n(d-1)}{2}$ edges, and by \Cref{clm:g_t-is-k-conn}, each of its $r+1$ components is $k$-connected.
    The graph $G$ itself is $\eps$-far from being $k$-connected, because it has $r+1 \ge \frac{\eps n d}{k}$ connected components.

    \paragraph{Query lower bound.}
    Let $\calA$ be a one-sided error $\eps$-tester for $k$-connectedness, and fix its random coins.
    Suppose that, on input $G$, $\calA$ makes fewer than $\frac{k}{32\eps}$ queries.
    We say that an edge $\{u,w\}$ is \emph{seen} if $\calA$ queried the position of $w$ in the adjacency list of $u$ or the position of $u$ in the adjacency list of $w$. 
    We will construct a $k$-connected $n$-vertex graph $G^*$ with the same degrees as $G$, on which every query of $\calA$ receives the same answer as on $G$.
    Then $\calA$ (with these coins) behaves identically on $G$ and on $G^*$, and since it must accept $G^*$, it accepts $G$.
    If this were the case for every choice of coins, then $\calA$ would accept $G$ with probability $1$, contradicting the fact that $G$ is $\eps$-far from $k$-connected.

    A supernode is \emph{exhausted} if all of its internal edges are seen.
    Each supernode has $\binom{d}{2}-h \ge \frac{d(d-1)}{2}-\frac{d-2}{2} \geq \frac{d^2}{4}$ internal edges, since $h \le \frac{k+1}{2}\le\frac{d-2}{2}$.
    So fewer than $\frac{4}{d^2}\cdot\frac{k}{32\eps} = \frac{k}{8\eps d^2}$ supernodes are exhausted.
    Since every component has at least $N_0 > \frac{k}{4\eps d^2}-1$ supernodes, there are more than
    \[
        \frac{k}{4\eps d^2} - 1 - \frac{k}{8\eps d^2} \;=\; \frac{k}{8\eps d^2}-1 \;\ge\; 4k-1 \;\ge\; 2h - 1
    \]
    non-exhausted supernodes in every component, that is, at least $2h$ of them.

    \emph{Constructing $G^*$.}
    Order the components of $G$ cyclically as $C_0, C_1, \ldots, C_r$, with indices modulo $r+1$.
    In every component $C_j$, choose $2h$ non-exhausted supernodes and, in each of them, one unseen internal edge.
    Assign $h$ of the chosen edges of $C_j$ to the pair $(C_{j-1},C_j)$ and the other $h$ edges to the pair $(C_j, C_{j+1})$.
    For every $j$ and every $\ell\in[h]$, take the $\ell$-th unseen edge $\{a,b\}$ of $C_j$ assigned to the pair $(C_j,C_{j+1})$ and the $\ell$-th unseen edge $\{a',b'\}$ of $C_{j+1}$ assigned to the same pair, and replace them by the edges $\{a,a'\}$ and $\{b,b'\}$. 
    Let $G^*$ be the resulting graph.

    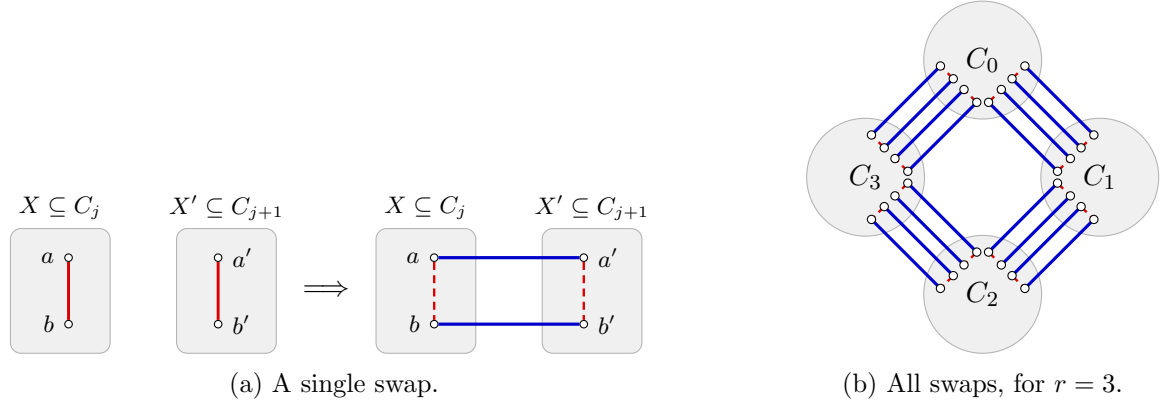
\begin{figure}[h]
        \centering
        \tikzset{
            intv/.style={circle, draw, fill=white, inner sep=1pt},
            removede/.style={draw=red!85!black, densely dashed, line width=0.9pt},
            unseene/.style={draw=red!85!black, line width=1.1pt},
            betweene/.style={draw=blue!80!black, line width=1.2pt}
        }
        \begin{subfigure}{0.56\textwidth}
            \centering
            \begin{tikzpicture}[scale=1.1]
                \begin{scope}
                    \draw[rounded corners, fill=black!6, draw=black!30] (-1.6,-0.75) rectangle (-0.4,0.75);
                    \draw[rounded corners, fill=black!6, draw=black!30] (0.4,-0.75) rectangle (1.6,0.75);
                    \node[font=\footnotesize] at (-1.0,1.0) {$X\subseteq C_j$};
                    \node[font=\footnotesize] at (1.0,1.0) {$X'\subseteq C_{j+1}$};
                    \node[intv, label=left:{\footnotesize $a$}] (xa) at (-0.9,0.4) {};
                    \node[intv, label=left:{\footnotesize $b$}] (xb) at (-0.9,-0.4) {};
                    \node[intv, label=right:{\footnotesize $a'$}] (xc) at (0.9,0.4) {};
                    \node[intv, label=right:{\footnotesize $b'$}] (xd) at (0.9,-0.4) {};
                    \draw[unseene] (xa) -- (xb);
                    \draw[unseene] (xc) -- (xd);
                \end{scope}
                \node at (2.2,0) {$\Longrightarrow$};
                \begin{scope}[xshift=4.4cm]
                    \draw[rounded corners, fill=black!6, draw=black!30] (-1.6,-0.75) rectangle (-0.4,0.75);
                    \draw[rounded corners, fill=black!6, draw=black!30] (0.4,-0.75) rectangle (1.6,0.75);
                    \node[font=\footnotesize] at (-1.0,1.0) {$X\subseteq C_j$};
                    \node[font=\footnotesize] at (1.0,1.0) {$X'\subseteq C_{j+1}$};
                    \node[intv, label=left:{\footnotesize $a$}] (ya) at (-0.9,0.4) {};
                    \node[intv, label=left:{\footnotesize $b$}] (yb) at (-0.9,-0.4) {};
                    \node[intv, label=right:{\footnotesize $a'$}] (yc) at (0.9,0.4) {};
                    \node[intv, label=right:{\footnotesize $b'$}] (yd) at (0.9,-0.4) {};
                    \draw[removede] (ya) -- (yb);
                    \draw[removede] (yc) -- (yd);
                    \draw[betweene] (ya) -- (yc);
                    \draw[betweene] (yb) -- (yd);
                \end{scope}
            \end{tikzpicture}
            \caption{A single swap.}
            \label{fig:swap}
        \end{subfigure}
        \hfill
        \begin{subfigure}{0.40\textwidth}
            \centering
            \begin{tikzpicture}[scale=1.0]
                \draw[fill=black!6, draw=black!30] (0.000,1.550) circle (0.78);
                \node at (0.000,1.550) {$C_0$};
                \draw[fill=black!6, draw=black!30] (1.550,0.000) circle (0.78);
                \node at (1.550,0.000) {$C_1$};
                \draw[fill=black!6, draw=black!30] (0.000,-1.550) circle (0.78);
                \node at (0.000,-1.550) {$C_2$};
                \draw[fill=black!6, draw=black!30] (-1.550,0.000) circle (0.78);
                \node at (-1.550,0.000) {$C_3$};
                \draw[betweene] (0.078,0.991) -- (0.991,0.078);
                \draw[betweene] (0.247,1.161) -- (1.161,0.247);
                \draw[betweene] (0.389,1.303) -- (1.303,0.389);
                \draw[betweene] (0.559,1.472) -- (1.472,0.559);
                \draw[removede] (0.078,0.991) -- (0.247,1.161);
                \draw[removede] (0.389,1.303) -- (0.559,1.472);
                \draw[removede] (0.991,0.078) -- (1.161,0.247);
                \draw[removede] (1.303,0.389) -- (1.472,0.559);
                \draw[betweene] (0.991,-0.078) -- (0.078,-0.991);
                \draw[betweene] (1.161,-0.247) -- (0.247,-1.161);
                \draw[betweene] (1.303,-0.389) -- (0.389,-1.303);
                \draw[betweene] (1.472,-0.559) -- (0.559,-1.472);
                \draw[removede] (0.991,-0.078) -- (1.161,-0.247);
                \draw[removede] (1.303,-0.389) -- (1.472,-0.559);
                \draw[removede] (0.078,-0.991) -- (0.247,-1.161);
                \draw[removede] (0.389,-1.303) -- (0.559,-1.472);
                \draw[betweene] (-0.078,-0.991) -- (-0.991,-0.078);
                \draw[betweene] (-0.247,-1.161) -- (-1.161,-0.247);
                \draw[betweene] (-0.389,-1.303) -- (-1.303,-0.389);
                \draw[betweene] (-0.559,-1.472) -- (-1.472,-0.559);
                \draw[removede] (-0.078,-0.991) -- (-0.247,-1.161);
                \draw[removede] (-0.389,-1.303) -- (-0.559,-1.472);
                \draw[removede] (-0.991,-0.078) -- (-1.161,-0.247);
                \draw[removede] (-1.303,-0.389) -- (-1.472,-0.559);
                \draw[betweene] (-0.991,0.078) -- (-0.078,0.991);
                \draw[betweene] (-1.161,0.247) -- (-0.247,1.161);
                \draw[betweene] (-1.303,0.389) -- (-0.389,1.303);
                \draw[betweene] (-1.472,0.559) -- (-0.559,1.472);
                \draw[removede] (-0.991,0.078) -- (-1.161,0.247);
                \draw[removede] (-1.303,0.389) -- (-1.472,0.559);
                \draw[removede] (-0.078,0.991) -- (-0.247,1.161);
                \draw[removede] (-0.389,1.303) -- (-0.559,1.472);
                \node[intv, inner sep=1.1pt] at (0.078,0.991) {};
                \node[intv, inner sep=1.1pt] at (0.247,1.161) {};
                \node[intv, inner sep=1.1pt] at (0.389,1.303) {};
                \node[intv, inner sep=1.1pt] at (0.559,1.472) {};
                \node[intv, inner sep=1.1pt] at (0.991,0.078) {};
                \node[intv, inner sep=1.1pt] at (1.161,0.247) {};
                \node[intv, inner sep=1.1pt] at (1.303,0.389) {};
                \node[intv, inner sep=1.1pt] at (1.472,0.559) {};
                \node[intv, inner sep=1.1pt] at (0.991,-0.078) {};
                \node[intv, inner sep=1.1pt] at (1.161,-0.247) {};
                \node[intv, inner sep=1.1pt] at (1.303,-0.389) {};
                \node[intv, inner sep=1.1pt] at (1.472,-0.559) {};
                \node[intv, inner sep=1.1pt] at (0.078,-0.991) {};
                \node[intv, inner sep=1.1pt] at (0.247,-1.161) {};
                \node[intv, inner sep=1.1pt] at (0.389,-1.303) {};
                \node[intv, inner sep=1.1pt] at (0.559,-1.472) {};
                \node[intv, inner sep=1.1pt] at (-0.078,-0.991) {};
                \node[intv, inner sep=1.1pt] at (-0.247,-1.161) {};
                \node[intv, inner sep=1.1pt] at (-0.389,-1.303) {};
                \node[intv, inner sep=1.1pt] at (-0.559,-1.472) {};
                \node[intv, inner sep=1.1pt] at (-0.991,-0.078) {};
                \node[intv, inner sep=1.1pt] at (-1.161,-0.247) {};
                \node[intv, inner sep=1.1pt] at (-1.303,-0.389) {};
                \node[intv, inner sep=1.1pt] at (-1.472,-0.559) {};
                \node[intv, inner sep=1.1pt] at (-0.991,0.078) {};
                \node[intv, inner sep=1.1pt] at (-1.161,0.247) {};
                \node[intv, inner sep=1.1pt] at (-1.303,0.389) {};
                \node[intv, inner sep=1.1pt] at (-1.472,0.559) {};
                \node[intv, inner sep=1.1pt] at (-0.078,0.991) {};
                \node[intv, inner sep=1.1pt] at (-0.247,1.161) {};
                \node[intv, inner sep=1.1pt] at (-0.389,1.303) {};
                \node[intv, inner sep=1.1pt] at (-0.559,1.472) {};
            \end{tikzpicture}
            \caption{All swaps, for $r=3$.}
            \label{fig:repair-cycle}
        \end{subfigure}
        \caption{The repair procedure in the proof of \Cref{thm:one-sided-lb}, for $k=3$, so that $h=2$.
        (a) Unseen internal edges $\{a,b\}$ of a supernode $X$ of $C_j$ and $\{a',b'\}$ of a supernode $X'$ of $C_{j+1}$ (red) are replaced by the edges $\{a,a'\}$ and $\{b,b'\}$ (blue); all degrees and all queried positions are unchanged.
        (b) The components $C_0,\ldots,C_r$ are arranged in a cycle, and $h$ swaps are performed between every two consecutive components, creating $h$ new edges between them.}
        \label{fig:repair}
    \end{figure}


    It remains to show that $G^*$ is $k$-connected. Consider a cut $(A,B)$ of $G^*$.
    For each component $C_j$, the graph $G^*$ contains a copy of $G_T(N_0,d,k)$ or $G_T(N_1,d,k)$ on the vertices of $C_j$, with at most one internal edge removed from each supernode.
    If the cut separates two vertices of the same component, then at least $k$ edges cross it, by \Cref{clm:g_t-is-k-conn} applied to this subgraph.
    Otherwise, $A$ and $B$ are unions of components of $G$.
    Since the components are arranged in a cycle, there are at least two indices $j$ such that $C_j$ and $C_{j+1}$ are on different sides of the cut, and for each such $j$, the $h$ new edges created for the pair $(C_j,C_{j+1})$ cross the cut.
    So at least $2h\ge k$ edges cross the cut.
    Therefore, $G^*$ is $k$-connected, which completes the proof.
\end{proof}

\bibliography{references}

@STRING{focs = {Proceedings of IEEE Symposium on Foundations of Computer Science (FOCS)}}

@STRING{soda = {Proceedings of ACM-SIAM Symposium on Discrete Algorithms (SODA)}}

@inproceedings{BermanRY14,
  author    = {Piotr Berman and
               Sofya Raskhodnikova and
               Grigory Yaroslavtsev},
  title     = {{$L_p$}-testing},
  booktitle = {Symposium on Theory of Computing, {STOC} 2014, New York, NY, USA,
               May 31 - June 03, 2014},
  pages     = {164--173},
  year      = {2014},
}

@article{OR11,
  author       = {Yaron Orenstein and
                  Dana Ron},
  title        = {Testing Eulerianity and connectivity in directed sparse graphs},
  journal      = {Theor. Comput. Sci.},
  volume       = {412},
  number       = {45},
  pages        = {6390--6408},
  year         = {2011},
  url          = {https://doi.org/10.1016/j.tcs.2011.06.038},
  doi          = {10.1016/J.TCS.2011.06.038}
}

@inproceedings{FNSYY20,
  author    = {Sebastian Forster and
               Danupon Nanongkai and
               Thatchaphol Saranurak and
               Liu Yang and
               Sorrachai Yingchareonthawornchai},
  editor    = {Shuchi Chawla},
  title     = {Computing and Testing Small Connectivity in Near-Linear Time and Queries
               via Fast Local Cut Algorithms},
  booktitle = soda,
  pages     = {2046--2065},
  publisher = {{SIAM}},
  year      = {2020}
}

@article{GR02,
  author    = {Oded Goldreich and
               Dana Ron},
  title     = {Property Testing in Bounded Degree Graphs},
  journal   = {Algorithmica},
  volume    = {32},
  number    = {2},
  pages     = {302--343},
  year      = {2002},
}

@article{PR02,
  author    = {Michal Parnas and
               Dana Ron},
  title     = {Testing the diameter of graphs},
  journal   = {Random Struct. Algorithms},
  volume    = {20},
  number    = {2},
  pages     = {165--183},
  year      = {2002},
}

@article{RS06,
  author    = {Sofya Raskhodnikova and
               Adam D. Smith},
  title     = {A Note on Adaptivity in Testing Properties of Bounded Degree Graphs},
  journal   = {Electronic Colloquium on Computational Complexity {(ECCC)}},
  volume    = {13},
  number    = {089},
  year      = {2006},
}

@inproceedings{Yao77,
  author    = {Andrew Chi{-}Chih Yao},
  title     = {Probabilistic Computations: Toward a Unified Measure of Complexity
               (Extended Abstract)},
  booktitle = focs,
  pages     = {222--227},
  year      = {1977},
}

@article{YoshidaI10,
  author       = {Yuichi Yoshida and
                  Hiro Ito},
  title        = {Testing \emph{k}-edge-connectivity of digraphs},
  journal      = {J. Syst. Sci. Complex.},
  volume       = {23},
  number       = {1},
  pages        = {91--101},
  year         = {2010},
  url          = {https://doi.org/10.1007/s11424-010-9280-5},
  doi          = {10.1007/S11424-010-9280-5},
  bibsource    = {dblp computer science bibliography, https://dblp.org}
}

@article{YoshidaI12,
  author       = {Yuichi Yoshida and
                  Hiro Ito},
  title        = {Property Testing on k-Vertex-Connectivity of Graphs},
  journal      = {Algorithmica},
  volume       = {62},
  number       = {3-4},
  pages        = {701--712},
  year         = {2012},
  url          = {https://doi.org/10.1007/s00453-010-9477-y},
  doi          = {10.1007/S00453-010-9477-Y},
  bibsource    = {dblp computer science bibliography, https://dblp.org}
}

@article{LeviPRV22,
  author       = {Amit Levi and
                  Ramesh Krishnan S. Pallavoor and
                  Sofya Raskhodnikova and
                  Nithin Varma},
  title        = {Erasure-Resilient Sublinear-Time Graph Algorithms},
  journal      = {{ACM} Trans. Comput. Theory},
  volume       = {14},
  number       = {1},
  pages        = {1:1--1:22},
  year         = {2022},
  url          = {https://doi.org/10.1145/3488250},
  doi          = {10.1145/3488250},
  bibsource    = {dblp computer science bibliography, https://dblp.org}
}

@article{BenderR02,
  author       = {Michael A. Bender and
                  Dana Ron},
  title        = {Testing properties of directed graphs: acyclicity and connectivity},
  journal      = {Random Struct. Algorithms},
  volume       = {20},
  number       = {2},
  pages        = {184--205},
  year         = {2002},
  url          = {https://doi.org/10.1002/rsa.10023},
  doi          = {10.1002/RSA.10023},
  bibsource    = {dblp computer science bibliography, https://dblp.org}
}
\bibliographystyle{alpha}

\appendix
\section{Proof of \Cref{lem:localec}}\label{app:localec}

The algorithm, given as Algorithm~\ref{alg:localec}, is the algorithm of~\cite[Algorithm~1]{FNSYY20} with the parameters used for \Cref{lem:localec}.
It repeatedly grows a depth-first search (DFS) tree from $v$, stops it at a random point, and reverses the edges on the path from $v$ to the vertex at which it stopped.
An edge is \emph{marked} when it is explored for the first time, and it stays marked (even after being reversed).

\begin{algorithm}[h]
	\caption{$\localec(v,k,\rho)$, adapted from~\cite[Algorithm~1]{FNSYY20}}
	\begin{algorithmic}[1]
		\Require vertex $v$; integer $k\ge1$; real $\rho$ with $1\le\rho\le\frac{m}{130k}$
		\State $\tau \gets \ceil{128k\rho}$
		\Loop~$k$ times:
		\State $y \gets \textsc{nil}$
		\State Grow a DFS tree $T$ of the vertices reachable from $v$ in the current graph, as follows:
		\While{the DFS has an edge $e=(a,b)$ to explore}
		\If{$e$ is not marked}
		\State Mark $e$
		\If{at least $\tau$ edges are marked} \Return $\bot$ \label{line:budget}
		\EndIf
		\State With probability $\frac{1}{8\rho}$, set $y\gets a$ and stop the DFS
		\EndIf
		\EndWhile
		\If{$y = \textsc{nil}$} \Return $S = V(T)$ \label{line:return-set}
		\EndIf
		\State Reverse the direction of every edge on the path from $v$ to $y$ in $T$
		\EndLoop
		\State \Return $\bot$ \label{line:return-bot}
	\end{algorithmic}
\label{alg:localec}
\end{algorithm}

We use the following observation (see~\cite[Lemma~3.5]{FNSYY20}).

\begin{lemma}\label{lem:path-reversal}
    Let $S\subseteq V$ with $v\in S$, and let $P$ be a path from $v$ to a vertex $y$. Reversing all edges of $P$ decreases $\dout(S)$ by exactly one if $y\notin S$, and does not change $\dout(S)$ if $y\in S$.
\end{lemma}

\begin{proof}
    Each edge of $P$ that crosses the cut $(S,V\setminus S)$ goes either from $S$ to $V\setminus S$ or from $V\setminus S$ to $S$. Since $P$ starts in $S$, the number of crossings of the first type exceeds the number of crossings of the second type by one if $y\notin S$, and the two numbers are equal if $y\in S$. Reversing $P$ swaps the two types.
\end{proof}

\begin{proof}[Proof of \Cref{lem:localec}]
    \emph{Complexity.}
    At most $\tau = O(k\rho)$ edges are marked. Each marked edge is obtained with one out-neighbor query, and each vertex visited by the DFS requires one out-degree query. Since every visited vertex other than $v$ is the head of a marked edge, the algorithm makes $O(k\rho)$ queries. Reversed edges are marked, so their current direction is known to the algorithm without further queries.
    As shown in~\cite[Lemma~3.4]{FNSYY20}, with suitable data structures, each of the $k$ iterations takes time linear in the number of marked edges, so the running time is $O(k^2\rho)$.

    \emph{Output.}
    Suppose that the algorithm returns $S = V(T)$ in Line~\ref{line:return-set} of iteration $r\le k$. Then $v \in S$, and no edge leaves $S$ in the current graph.
    All out-edges of vertices of $S$ in the current graph were explored by the DFS and are thus marked. Fewer than $\tau$ edges are marked, and $\tau \le \ceil{\frac{128m}{130}} \le m$, so the vertices of $S$ have fewer than $m$ out-edges in total, and therefore $S\neq V$.
    Before iteration $r$, exactly $r-1\le k-1$ paths starting at $v$ were reversed. By \Cref{lem:path-reversal}, each of them decreased $\dout(S)$ by at most one, so $\dout(S)\le k-1$ in the original graph.

    \emph{Success probability.}
    Suppose that there exist $\gamma\in[k]$ and $S'\ni v$ with $\vol(S')\le\rho\gamma$ and $\dout(S')\le k-\gamma$.
    The algorithm returns $\bot$ only in one of the following two ways.
    
    First, it may reach the budget in Line~\ref{line:budget}. Every newly marked edge, except the one that reaches the budget, triggers an independent coin flip with success probability $p = \frac{1}{8\rho}$, and every success ends an iteration. Thus, Line~\ref{line:budget} is reached only if fewer than $k$ of the first $\tau-1$ coin flips succeed, that is, only if the number $Y$ of coin flips needed for $k$ successes is at least $\tau$. Since $\E[Y] = k/p = 8k\rho$, by Markov's inequality, this happens with probability at most $\frac{8k\rho}{\tau}\le\frac{1}{16}$.

    Second, it may complete all $k$ iterations, setting $y$ in each of them, and return $\bot$ in Line~\ref{line:return-bot}.
    Let $Z$ be the number of iterations in which $y$ is set to a vertex of $S'$.
    By \Cref{lem:path-reversal}, each iteration in which $y\notin S'$ decreases $\dout(S')$ by one. Moreover, once $\dout(S') = 0$ in the current graph, no vertex outside $S'$ is reachable from $v$, so $y\in S'$ in all later iterations.
    Hence, at most $k-\gamma$ iterations have $y\notin S'$, and returning $\bot$ in Line~\ref{line:return-bot} implies $Z\ge\gamma$.
    Now, $y$ is set to a vertex of $S'$ only when the coin flip of a newly marked edge $(a,b)$ with $a\in S'$ succeeds. A newly marked edge has never been reversed, since only marked edges are reversed, so $(a,b)$ is one of the $\vol(S')\le\rho\gamma$ out-edges of vertices of $S'$ in the original graph. Each edge is newly marked at most once, so $\E[Z]\le \rho\gamma\cdot p = \frac{\gamma}{8}$, and by Markov's inequality, $\Pr[Z\ge\gamma]\le\frac18$.

    Therefore, the algorithm returns $\bot$ with probability at most $\frac{1}{16}+\frac{1}{8}<\frac14$.
\end{proof}

\section{Testing when $\eps\davg$ is large}\label{app:large-eps}
 
In this section, we show that testing $k$-edge-connectedness of simple graphs is easy when $\eps\davg$ is at least a sufficiently large constant: it suffices to check the degrees of a few random vertices.
 
\begin{lemma}\label{lem:large-eps}
    The following algorithm is a one-sided error $\eps$-tester for $k$-edge-connectedness of simple directed graphs with average degree $\davg \ge \frac{2}{\eps}$: sample $s = \ceil{\frac{2k\ln 3}{\eps\davg}}$ vertices uniformly at random, and reject if and only if one of them has out-degree or in-degree less than $k$.
    It makes $O\big(\frac{k}{\eps\davg}\big)$ queries.
    For simple undirected graphs with $\davg \ge \frac{4}{\eps}$, the same algorithm works with $s = \ceil{\frac{4k\ln 3}{\eps\davg}}$.
\end{lemma}
 
\begin{proof}
    If $G$ is $k$-edge-connected, every vertex has out- and in-degree at least $k$, so the algorithm always accepts.
 
    Now let $G=(V,E)$ be a simple directed graph with $m$ edges that is $\eps$-far from $k$-edge-connected.
    By \Cref{fact:far-characterization}, there is a family $\cF$ of pairwise disjoint nonempty proper subsets of $V$ such that $\sum_{S\in\cF}(k-\dout(S))>\eps m$ or $\sum_{S\in\cF}(k-\din(S))>\eps m$.
    Since the algorithm checks both out- and in-degrees, we assume the former holds wlog.
    Remove from $\cF$ all sets $S$ with $\dout(S)\ge k$.
    As shown in the proof of \Cref{clm:exp-k-quality}, every set in $\cF$ whose vertices all have out-degree at least $k$ has at least $k+1$ vertices.
    Since the sets in $\cF$ are disjoint, there are at most $\frac{n}{k+1}$ such sets, and they contribute at most $\frac{kn}{k+1} < n \le \frac{\eps m}{2}$ to the sum, where the last inequality follows from $m = \davg n$ and $\eps\davg\ge 2$.
    Therefore, the remaining sets contribute more than $\frac{\eps m}{2}$ to the sum.
    Each of them contributes at most $k$ and contains a vertex of out-degree less than $k$, so more than $\frac{\eps m}{2k} = \frac{\eps\davg}{2k}\cdot n$ vertices have out-degree less than $k$.
    Hence, the probability that all $s$ sampled vertices have out-degree at least $k$ is at most $\left(1-\frac{\eps\davg}{2k}\right)^{s} \le e^{-\ln 3} = \frac13$, and the algorithm rejects with probability at least $\frac23$.
    The algorithm makes $2s = O\big(\frac{k}{\eps\davg}\big)$ degree queries.
 
    For a simple undirected graph $G$ with $\eps\davg\ge 4$, consider its directed version $\Gdir$, which is simple, has average out-degree $\davg$, and in which the out- and in-degree of every vertex equal its degree in $G$.
    Running the directed algorithm on $\Gdir$ with parameter $\eps/2$, which satisfies $\frac{\eps}{2}\davg\ge 2$, gives exactly the undirected algorithm, and its correctness follows from \Cref{lem:reduction-completeness,lem:reduction-soundness}.
\end{proof}

\end{document}